\documentclass[journal]{IEEEtran}

\usepackage{adjustbox}
\usepackage{algorithm}
\usepackage{algpseudocode}
\usepackage{amsfonts}
\usepackage{amsmath}
\usepackage{amssymb}
\usepackage{amsthm}
\usepackage{array}
\usepackage{cite}
\usepackage{colortbl}
\usepackage{environ}
\usepackage{graphicx}
\usepackage{grffile}
\usepackage{hyperref}
\usepackage{import}
\usepackage{mathtools}
\usepackage{microtype}
\usepackage{multirow}
\usepackage{pgfplots}
\usepackage{siunitx}
\usepackage{stfloats}
\usepackage{tikz}
\usepackage{url}
\usepackage{xcolor}
\usepackage[RPvoltages]{circuitikz}
\usepackage[T1]{fontenc}
\usepackage[caption=false,font=footnotesize,subrefformat=parens,labelformat=parens]{subfig}
\usepackage[short]{optidef}
\usepackage[subtle]{savetrees}

\pgfplotsset{compat=newest}
\usetikzlibrary{plotmarks}
\usetikzlibrary{arrows.meta}
\usepgfplotslibrary{patchplots}
\newtheorem{proposition}{Proposition}
\newtheorem{remark}{Remark}
\DeclareSIUnit{\belm}{Bm}
\DeclareSIUnit{\dBm}{\deci\belm}
\DeclareSIUnit{\beli}{Bi}
\DeclareSIUnit{\dBi}{\deci\beli}
\DeclareSIUnit{\belA}{BA}
\DeclareSIUnit{\dBA}{\deci\belA}
\usetikzlibrary{arrows,matrix,positioning}

\algrenewcommand{\algorithmicwhile}{\textbf{While}}
\algrenewcommand{\algorithmicif}{\textbf{If}}
\algrenewcommand{\algorithmicthen}{\textbf{Then}}
\algrenewcommand{\algorithmicelse}{\textbf{Else}}
\algrenewcommand{\algorithmicend}{\textbf{End}}
\algrenewcommand{\algorithmicrepeat}{\textbf{Repeat}}
\algrenewcommand{\algorithmicuntil}{\textbf{Until}}

\begin{document}
	\bstctlcite{IEEEexample:BSTcontrol}
	\title{IRS-Aided SWIPT:\\Joint Waveform, Active and Passive Beamforming Design Under Nonlinear Harvester Model}
	\author{
		\IEEEauthorblockN{
			Yang~Zhao,~\IEEEmembership{Member,~IEEE,}
			~Bruno~Clerckx,~\IEEEmembership{Senior~Member,~IEEE,}
			and~Zhenyuan~Feng,~\IEEEmembership{Member,~IEEE}
		}
		\thanks{
			The authors are with the Department of Electrical and Electronic Engineering, Imperial College London, London SW7 2AZ, U.K. (e-mail: \{yang.zhao18, b.clerckx, zhenyuan.feng19\}@imperial.ac.uk).

			Digital Object Identifier 10.1109/TCOMM.2021.3129931.
		}
	}
	\maketitle

	\begin{abstract}
		The performance of Simultaneous Wireless Information and Power Transfer (SWIPT) is mainly constrained by the received Radio-Frequency (RF) signal strength. To tackle this problem, we introduce an Intelligent Reflecting Surface (IRS) to compensate the propagation loss and boost the transmission efficiency. This paper proposes a novel IRS-aided SWIPT system where a multi-carrier multi-antenna Access Point (AP) transmits information and power simultaneously, with the assist of an IRS, to a single-antenna User Equipment (UE) employing practical receiving schemes. Considering harvester nonlinearity, we characterize the achievable Rate-Energy (R-E) region through a joint optimization of waveform, active and passive beamforming based on the Channel State Information at the Transmitter (CSIT). This problem is solved by the Block Coordinate Descent (BCD) method, where we obtain the active precoder in closed form, the passive beamforming by the Successive Convex Approximation (SCA) approach, and the waveform amplitude by the Geometric Programming (GP) technique. To facilitate practical implementation, we also propose a low-complexity design based on closed-form adaptive waveform schemes. Simulation results demonstrate the proposed algorithms bring considerable R-E gains with robustness to CSIT inaccuracy and finite IRS states, and emphasize the importance of modeling harvester nonlinearity in the IRS-aided SWIPT design.
	\end{abstract}


	\begin{section}{Introduction}
		\begin{subsection}{Simultaneous Wireless Information and Power Transfer}
			\IEEEPARstart{W}{ith} the great advance in communication performance, a bottleneck of wireless networks has come to energy supply. Simultaneous Wireless Information and Power Transfer (SWIPT) is a promising solution to connect and power mobile devices via Radio-Frequency (RF) waves. It provides low power at \si{\uW} level but broad coverage up to hundreds of meters in a sustainable and controllable manner, bringing more opportunities to the Internet of Things (IoT) and Machine to Machine (M2M) networks. The upsurge in wireless devices, together with the decrease of electronics power consumption, calls for a re-thinking of future wireless networks based on Wireless Power Transfer (WPT) and SWIPT \cite{Clerckx2019}.

            The concept of SWIPT was first cast in \cite{Varshney2008}, where the authors investigated the Rate-Energy (R-E) tradeoff for a flat Gaussian channel and typical discrete channels. \cite{Zhou2013} proposed two practical co-located information and power receivers, i.e., Time Switching (TS) and Power Splitting (PS). Dedicated information and energy beamforming were then investigated in \cite{Zhang2013,Park2014} to characterize the R-E region for multi-antenna broadcast and interference channels. On the other hand, \cite{Trotter2009} pointed out that the RF-to-DC conversion efficiency of rectifiers depends on the input power and waveform shape. It implies that the modeling of the energy harvester, particularly its nonlinearity, has a crucial impact on the waveform preference, resource allocation, and system design of any wireless-powered systems \cite{Trotter2009,Clerckx2018,Clerckx2019}. Motivated by this, \cite{Clerckx2016a} derived a tractable nonlinear harvester model based on the Taylor expansion of diode I-V characteristics, and performed joint waveform and beamforming design for WPT. Simulation and experiments showed the benefit of modeling energy harvester nonlinearity in real system design \cite{Kim2019,Kim2020a} and demonstrated the joint waveform and beamforming strategy as a key technique to expand the operation range \cite{Kim2021}. A low-complexity adaptive waveform design by Scaled Matched Filter (SMF) was proposed in \cite{Clerckx2017} to exploit the rectifier nonlinearity, whose advantage was then demonstrated in a prototype with channel acquisition \cite{Kim2017}. Beyond WPT, \cite{Clerckx2018b} uniquely showed that the rectifier nonlinearity brings radical changes to SWIPT design, namely (i) modulated and unmodulated waveforms are not equally suitable for wireless power delivery; (ii) a multi-carrier unmodulated waveform superposed to a multi-carrier modulated waveform can enlarge the R-E region; (iii) a combination of PS and TS is generally the best strategy; (iv) the optimal input distribution is not the conventional Circularly Symmetric Complex Gaussian (CSCG); (v) modeling rectifier nonlinearity is beneficial to system performance and essential to efficient SWIPT design. Those observations, validated experimentally in \cite{Kim2019}, led to the question: \emph{What is the optimal input distribution for SWIPT under nonlinearity?} This question was answered in \cite{Varasteh2020} for single-carrier SWIPT, and some attempts were further made in \cite{Varasteh2019d} for multi-carrier SWIPT. The answers shed new light to the fundamental limits of SWIPT and practical signaling (e.g., modulation and waveform) strategies. It is now well understood from \cite{Clerckx2018b,Varasteh2020,Varasteh2019d} that, due to harvester nonlinearity, a combination of CSCG and on-off keying in single-carrier setting and non-zero mean asymmetric inputs in multi-carrier setting lead to significantly larger R-E region compared to conventional CSCG. Recently, \cite{Varasteh2020a} used machine learning techniques to design SWIPT signaling under nonlinearity to complement the information-theoretic results of \cite{Varasteh2020}, and new modulation schemes were subsequently invented.
		\end{subsection}

		\begin{subsection}{Intelligent Reflecting Surface}
			Intelligent Reflecting Surface (IRS) has recently emerged as a promising technique that adapts the propagation environment to enhance the spectrum and energy efficiency. In practice, an IRS consists of multiple individual sub-wavelength reflecting elements to adjust the amplitude and phase of the incoming signal (i.e., passive beamforming). Different from the relay, backscatter and frequency-selective surface \cite{Anwar2018}, the IRS assists the primary transmission using passive components with negligible thermal noise but is limited to frequency-dependent reflection.

			Inspired by the development of real-time reconfigurable metamaterials \cite{Cui2014}, the authors of \cite{Liaskos2018} introduced a programmable metasurface that steers or polarizes the electromagnetic wave at a specific frequency to mitigate signal attenuation. \cite{Wu2018} proposed an IRS-assisted Multiple-Input Single-Output (MISO) system and jointly optimized the precoder at the Access Point (AP) and the phase shifts at the IRS to minimize the transmit power. The active and passive beamforming problem was then extended to the discrete phase shift case \cite{Wu2019a} and the multi-user case \cite{Wu2019}. In \cite{Abeywickrama2020}, the authors investigated the impact of non-zero resistance on the reflection pattern and emphasized the coupling between reflection amplitude and phase shift in practice. To estimate the cascaded AP-IRS-User Equipment (UE) link without RF-chains at the IRS, practical protocols were developed based on element-wise on/off switching \cite{Nadeem2019}, training sequence and reflection pattern design \cite{You2019,Kang2020}, and compressed sensing \cite{Wang2020}. The hardware architecture, design challenges, and application opportunities of practical IRS were covered in \cite{Wu2020}. In \cite{Dai2020}, a prototype IRS with \num{256} \num{2}-bit elements based on Positive Intrinsic-Negative (PIN) diodes was developed to support real-time video transmission at \si{GHz} and mmWave frequency.
		\end{subsection}

		\begin{subsection}{IRS-Aided SWIPT}
			By integrating IRS with SWIPT, the constructive reflection can boost the end-to-end power efficiency and improve the R-E tradeoff. In multi-user cases, dedicated energy beams were proved unnecessary for the Weighted Sum-Power (WSP) maximization \cite{Wu2020b} but essential when fairness issue is considered \cite{Tang2019}. It was also claimed that Line-of-Sight (LoS) links could boost the WSP since rank-deficient channels require fewer energy beams \cite{Wu2020a}. However, \cite{Wu2020b,Tang2019,Wu2020a} were based on a linear energy harvester model that is known in both the RF and the communication literature to be inefficient and inaccurate \cite{Clerckx2019,Trotter2009,Clerckx2018,Clerckx2016a,Kim2019,Kim2020a,Kim2021,Clerckx2017,Kim2017,Clerckx2018b,Varasteh2020,Varasteh2019d,Varasteh2020a}. Based on practical IRS and harvester models, \cite{Xu2021c} introduced a scalable resource allocation framework for a large-scale tile-based IRS-assisted SWIPT system, where the optimization consists of a reflection design stage and a joint reflection selection and precoder design stage. The proposed framework provides a flexible tradeoff between performance and complexity. To the best of our knowledge, all existing papers considered resource allocation and beamforming design for dedicated information and energy users in a single-carrier network. In this paper, we instead build our design based on a proper nonlinear harvester model that captures the dependency of the output DC power on both the power and shape of the input waveform, and marry the benefits of joint multi-carrier waveform and active beamforming optimization for SWIPT with the passive beamforming capability of IRS, to investigate the R-E tradeoff for one SWIPT user with co-located information decoder and energy harvester. We ask ourselves the important question: \emph{How to jointly exploit the spatial domain and the frequency domain efficiently through joint waveform and beamforming design to enlarge the R-E region of IRS-aided SWIPT?} The contributions of this paper are summarized as follows.

			\emph{First,} we propose a novel IRS-aided SWIPT architecture based on joint waveform, active and passive beamforming design under the diode nonlinear model \cite{Clerckx2016a}. Although this tractable harvester model accurately reveals how the input power level and waveform shape influence the output DC power, it also introduces design challenges such as frequency coupling (i.e., components of different frequencies compensate and produce DC), waveform coupling (i.e., different waveforms jointly contribute to DC), and high-order objective function. To make an efficient use of the rectifier nonlinearity, we superpose a multi-carrier unmodulated power waveform (deterministic multisine) to a multi-carrier modulated information waveform and evaluate the performance under the TS and PS receiving modes. The proposed joint waveform, active and passive beamforming architecture exploits the rectifier nonlinearity, the channel selectivity, and a beamforming gain across frequency and spatial domains to enlarge the achievable R-E region. This is the first paper to propose a joint waveform, active and passive beamforming architecture for IRS-aided SWIPT.

			\emph{Second,} we characterize each R-E boundary point by energy maximization under a rate constraint. The problem is solved by a Block Coordinate Descent (BCD) algorithm based on the Channel State Information at the Transmitter (CSIT). For active beamforming, we prove that the global optimal active information and power precoders coincide at Maximum-Ratio Transmission (MRT) even with rectifier nonlinearity. For passive beamforming, we propose a Successive Convex Approximation (SCA) algorithm and retrieve the IRS phase shift by eigen decomposition with optimality proof. Finally, the superposed waveform and the splitting ratio are optimized by the Geometric Programming (GP) technique. The IRS phase shift, active precoder, and waveform amplitude are updated iteratively until convergence. This is the first paper to jointly optimize waveform and active/passive beamforming in IRS-aided SWIPT.

			\emph{Third,} we introduce two closed-form adaptive waveform schemes to avoid the exponential complexity of the GP algorithm. To facilitate practical SWIPT implementation, the Water-Filling (WF) strategy for modulated waveform and the SMF strategy for multisine waveform are combined in time and power domains, respectively. The passive beamforming design is also adapted to accommodate the low-complexity waveform schemes. The proposed low-complexity BCD algorithm achieves a good balance between performance and complexity.

			\emph{Fourth,} we provide numerical results to evaluate the proposed algorithms. It is concluded that (i) IRS enables constructive reflection and flexible subchannel design in the frequency domain that is essential for SWIPT systems; (ii) IRS mainly affects the effective channel instead of the waveform design; (iii) multisine waveform is beneficial to energy transfer especially when the number of subbands is large; (iv) TS is preferred at low Signal-to-Noise Ratio (SNR) while PS is preferred at high SNR; (v) there exist two optimal IRS development locations, one close to the AP and one close to the UE; (vi) the output SNR scales linearly with the number of transmit antennas and quadratically with the number of IRS elements; (vii) due to the rectifier nonlinearity, the output DC scales quadratically with the number of transmit antennas and quartically with the number of IRS elements; (viii) for narrowband SWIPT, the optimal active and passive beamforming for any R-E point are also optimal for the whole R-E region; (ix) for broadband SWIPT, the optimal active and passive beamforming depend on specific R-E point and require adaptive designs; (x) the proposed algorithms are robust to practical impairments such as inaccurate cascaded CSIT and finite IRS reflection states.

			\emph{Organization:} Section~\ref{se:system_model} introduces the system model. Section~\ref{se:problem_formulation} formulates the problem and tackles the waveform, active and passive beamforming design. Section~\ref{se:performance_evaluation} provides simulation results. Section~\ref{se:conclusion_and_future_works} concludes the paper.

			\emph{Notations:} Scalars, vectors and matrices are denoted respectively by italic, bold lower-case, and bold upper-case letters. $j$ denotes the imaginary unit. $\boldsymbol{0}$ and $\boldsymbol{1}$ denote respectively zero and one vector or matrix. $\boldsymbol{I}$ denotes the identity matrix. $\mathbb{R}_+^{x \times y}$ and $\mathbb{C}^{x \times y}$ denote respectively the space of real nonnegative and complex $x \times y$ matrices. $\Re\{\cdot\}$ retrieves the real part of a complex entity. $[\cdot]_{(n)}$ denotes the $n$-th entry of a vector and $[\cdot]_{(1:n)}$ denotes the first $n$ entries of a vector. $(\cdot)^*$, $(\cdot)^T$, $(\cdot)^H$, $(\cdot)^+$, $\lvert{\cdot}\rvert$, $\lVert{\cdot}\rVert$ represent respectively the conjugate, transpose, conjugate transpose, ramp function, absolute value, and Euclidean norm. $\arg(\cdot)$, $\mathrm{rank}(\cdot)$, $\mathrm{tr}(\cdot)$, $\mathrm{diag}(\cdot)$ and $\mathrm{diag}^{-1}(\cdot)$ denote respectively the argument, rank, trace, a square matrix with input vector on the main diagonal, and a vector retrieving the main diagonal of the input matrix. $\odot$ denotes the Hadamard product. $\boldsymbol{S} \succeq \boldsymbol{0}$ means $\boldsymbol{S}$ is positive semi-definite. $\mathbb{A}\{\cdot\}$ extracts the DC component of a signal. $\mathbb{E}_X\{\cdot\}$ takes expectation w.r.t. random variable $X$ ($X$ is omitted for simplicity). The distribution of a CSCG random vector with mean $\boldsymbol{0}$ and covariance $\boldsymbol{\Sigma}$ is denoted by $\mathcal{CN}(\boldsymbol{0},\boldsymbol{\Sigma})$. $\sim$ means ``distributed as''. $(\cdot)^{(i)}$ and $(\cdot)^{\star}$ denote respectively the $i$-th iterated value and the stationary solution.
		\end{subsection}
	\end{section}

	\begin{section}{System Model}\label{se:system_model}
		\begin{figure}[!t]
			\centering
			\def\svgwidth{0.9\columnwidth}
\begingroup%
  \makeatletter%
  \providecommand\color[2][]{%
    \errmessage{(Inkscape) Color is used for the text in Inkscape, but the package 'color.sty' is not loaded}%
    \renewcommand\color[2][]{}%
  }%
  \providecommand\transparent[1]{%
    \errmessage{(Inkscape) Transparency is used (non-zero) for the text in Inkscape, but the package 'transparent.sty' is not loaded}%
    \renewcommand\transparent[1]{}%
  }%
  \providecommand\rotatebox[2]{#2}%
  \newcommand*\fsize{\dimexpr\f@size pt\relax}%
  \newcommand*\lineheight[1]{\fontsize{\fsize}{#1\fsize}\selectfont}%
  \ifx\svgwidth\undefined%
    \setlength{\unitlength}{252.90451594bp}%
    \ifx\svgscale\undefined%
      \relax%
    \else%
      \setlength{\unitlength}{\unitlength * \real{\svgscale}}%
    \fi%
  \else%
    \setlength{\unitlength}{\svgwidth}%
  \fi%
  \global\let\svgwidth\undefined%
  \global\let\svgscale\undefined%
  \makeatother%
  \begin{picture}(1,0.4981921)%
    \lineheight{1}%
    \setlength\tabcolsep{0pt}%
    \put(0,0){\includegraphics[width=\unitlength]{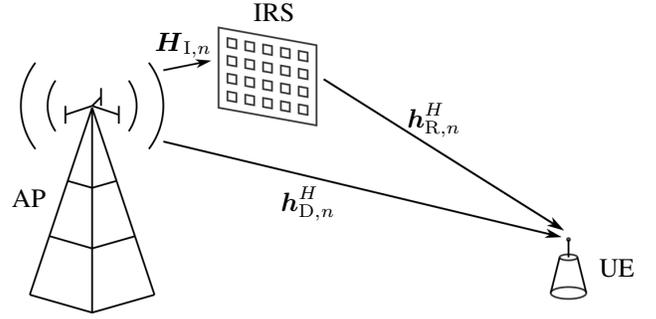}}%
    \put(-0.00020079,0.17648561){\color[rgb]{0,0,0}\makebox(0,0)[lt]{\lineheight{1.25}\smash{\begin{tabular}[t]{l}AP\end{tabular}}}}%
    \put(0.40247465,0.48648438){\color[rgb]{0,0,0}\makebox(0,0)[lt]{\lineheight{1.25}\smash{\begin{tabular}[t]{l}IRS\end{tabular}}}}%
    \put(0.97930298,0.05996617){\color[rgb]{0,0,0}\makebox(0,0)[lt]{\lineheight{1.25}\smash{\begin{tabular}[t]{l}UE\end{tabular}}}}%
    \put(0.44815942,0.171038){\color[rgb]{0,0,0}\makebox(0,0)[lt]{\lineheight{1.25}\smash{\begin{tabular}[t]{l}$\boldsymbol{h}_{\mathrm{D},n}^H$\end{tabular}}}}%
    \put(0.6588808,0.31057413){\color[rgb]{0,0,0}\makebox(0,0)[lt]{\lineheight{1.25}\smash{\begin{tabular}[t]{l}$\boldsymbol{h}_{\mathrm{R},n}^H$\end{tabular}}}}%
    \put(0.23848906,0.43809175){\color[rgb]{0,0,0}\makebox(0,0)[lt]{\lineheight{1.25}\smash{\begin{tabular}[t]{l}$\boldsymbol{H}_{\mathrm{I},n}$\end{tabular}}}}%
  \end{picture}%
\endgroup%

			\caption{An IRS-aided multi-carrier MISO SWIPT system.}
			\label{fi:system}
		\end{figure}

		As shown in Fig.~\ref{fi:system}, we propose an IRS-aided SWIPT system where an $M$-antenna AP delivers information and power simultaneously, through an $L$-element IRS, to a single-antenna UE over $N$ orthogonal evenly-spaced subbands. We consider a quasi-static block fading model and assume the CSIT of direct and cascaded channels is known. The signals reflected by two or more times are omitted, and the noise power is assumed too small to be harvested.

		\begin{subsection}{Transmitted Signal}
			Following \cite{Clerckx2018b}, we superpose a multi-carrier modulated information-bearing waveform to a multi-carrier unmodulated power-dedicated deterministic multisine to boost the spectrum and energy efficiency. The information signal transmitted over subband $n \in \{1, \dots, N\}$ at time $t$ is
			\begin{equation}
				\boldsymbol{x}_{\mathrm{I},n}(t) = \Re\left\{\boldsymbol{w}_{\mathrm{I},n} \tilde{x}_{\mathrm{I},n}(t) e^{j2{\pi}{f_n}{t}}\right\},
			\end{equation}
			where $\boldsymbol{w}_{\mathrm{I},n} \in \mathbb{C}^{M \times 1}$ is the information precoder at subband $n$, $\tilde{x}_{\mathrm{I},n}\sim\mathcal{CN}(0,1)$ is the information symbol at subband $n$, and $f_n$ is the frequency of subband $n$. On the other hand, the power signal transmitted over subband $n$ at time $t$ is
			\begin{equation}
				\boldsymbol{x}_{\mathrm{P},n}(t) = \Re\left\{\boldsymbol{w}_{\mathrm{P},n} e^{j2{\pi}{f_n}{t}}\right\},
			\end{equation}
			where $\boldsymbol{w}_{\mathrm{P},n} \in \mathbb{C}^{M \times 1}$ is the power precoder at subband $n$. Therefore, the superposed signal transmitted over all subbands at time $t$ is
			\begin{equation}
				\boldsymbol{x}(t) = \Re{\left\{\sum_{n=1}^N(\boldsymbol{w}_{\mathrm{I},n}\tilde{x}_{\mathrm{I},n}(t)+\boldsymbol{w}_{\mathrm{P},n}){e^{j2{\pi}{f_n}{t}}}\right\}}.
			\end{equation}
			We also define $\boldsymbol{w}_{\mathrm{I/P}} \triangleq [\boldsymbol{w}_{\mathrm{I/P},1}^T,\dots,\boldsymbol{w}_{\mathrm{I/P},N}^T]^T \in \mathbb{C}^{MN \times 1}$.
		\end{subsection}

		\begin{subsection}{Reflection Pattern and Composite Channel}
			According to Green's decomposition \cite{Hansen1989}, the backscattered signal of an antenna can be decomposed into the \emph{structural mode} component and the \emph{antenna mode} component. The former is fixed and can be regarded as part of the environment multipath, while the latter is adjustable and depends on the mismatch of the antenna and load impedance. IRS element $l \in \{1, \dots, L\}$ varies its impedance $Z_l = R_l + j X_l$ to reflect the incoming signal, and the reflection coefficient is defined as
			\begin{equation}
				\phi_l = \frac{Z_l - Z_0}{Z_l + Z_0} \triangleq \eta_l e^{j\theta_l},
			\end{equation}
			where $Z_0$ is the real-valued characteristic impedance, $\eta_l \in [0, 1]$ is the reflection amplitude,\footnote{Due to the non-zero power consumption at the IRS, $R_l > 0$ in practice such that $\eta_l < 1$ and is a function of $\theta_l$. This paper sticks to the most common IRS model where the reflection amplitude equals \num{1} so as to reduce the design complexity and provide a primary benchmark for practical IRS-aided SWIPT.} and $\theta_l \in [0,2\pi)$ is the phase shift. We also define $\boldsymbol{\phi} \triangleq [\phi_1, \dots, \phi_L]^H \in \mathbb{C}^{L \times 1}$ and $\boldsymbol{\Theta} \triangleq \mathrm{diag}(\phi_1, \dots, \phi_L)=\mathrm{diag}(\boldsymbol{\Phi}^*) \in \mathbb{C}^{L \times L}$ as the IRS vector and matrix, respectively.

			\begin{remark}\label{re:reflection_coefficient}
				The element impedance $Z_l$ maps to the reflection coefficient $\phi_l$ uniquely. Since the reactance $X_l$ depends on the frequency, the reflection coefficient $\phi_l$ is also a function of frequency and cannot be designed independently at different subbands. In this paper, we assume the bandwidth is small compared to the operating frequency such that the reflection coefficient of each IRS element is the same at all subbands.
			\end{remark}

			At subband $n$, we denote the AP-UE direct channel as $\boldsymbol{h}_{\mathrm{D},n}^H \in \mathbb{C}^{1 \times M}$, the AP-IRS incident channel as $\boldsymbol{H}_{\mathrm{I},n} \in \mathbb{C}^{L \times M}$, and the IRS-UE reflected channel as $\boldsymbol{h}_{\mathrm{R},n}^H \in \mathbb{C}^{1 \times L}$. The auxiliary AP-IRS-UE link can be modeled as a concatenation of the incident channel, the IRS reflection, and the reflected channel. Hence, the composite equivalent channel reduces to
			\begin{equation}\label{eq:h_n}
				\boldsymbol{h}_{n}^H = \boldsymbol{h}_{\mathrm{D},n}^H + \boldsymbol{h}_{\mathrm{R},n}^H \boldsymbol{\Theta} \boldsymbol{H}_{\mathrm{I},n} = \boldsymbol{h}_{\mathrm{D},n}^H + \boldsymbol{\phi}^H \boldsymbol{V}_{n},
			\end{equation}
			where we define the cascaded incident-reflected channel at subband $n$ as $\boldsymbol{V}_{n} \triangleq \mathrm{diag}(\boldsymbol{h}_{\mathrm{R},n}^H)\boldsymbol{H}_{\mathrm{I},n} \in \mathbb{C}^{L \times M}$. We also define $\boldsymbol{h} \triangleq [\boldsymbol{h}_1^T,\dots,\boldsymbol{h}_N^T]^T \in \mathbb{C}^{MN \times 1}$.

			\begin{remark}\label{re:subband_tradeoff}
				The cascaded channel varies at different subbands, but the reflection cannot be designed independently at different frequencies. Therefore, there exists a tradeoff for the passive beamforming design in the frequency domain, and the composite subchannels should be tuned adaptively to meet the specific requirement of multi-carrier SWIPT. For example, one can design the reflection pattern to either enhance the strongest subband (e.g., $\max_{\boldsymbol{\phi},n} \lVert \boldsymbol{h}_n \rVert$), or improve the fairness among subbands (e.g., $\max_{\boldsymbol{\phi}} \min_n \lVert \boldsymbol{h}_n \rVert$). That is to say, IRS essentially enables a flexible subchannel design. In the MISO case, a similar effect also exists in the spatial domain. Therefore, each reflection coefficient is indeed shared by $M$ antennas over $N$ subbands.
			\end{remark}
		\end{subsection}

		\begin{subsection}{Received Signal}
			The received superposed signal at the single-antenna UE is\footnote{We assume that the time difference of signal arrival via direct and auxiliary link is negligible compared to the symbol period.}
			\begin{equation}
				y(t) = \Re\left\{\sum_{n=1}^N{\Bigl(\boldsymbol{h}_{n}^H}{(\boldsymbol{w}_{\mathrm{I},n}\tilde{x}_{\mathrm{I},n}(t)+\boldsymbol{w}_{\mathrm{P},n})+\tilde{n}_n(t)\Bigr)}{e^{j2{\pi}{f_n}{t}}}\right\},
			\end{equation}
			where $\tilde{n}_n(t)$ is the noise at RF band $n$. Note that the modulated component can be used for energy harvesting if necessary, but the multisine component carries no information and cannot be used for information decoding.
		\end{subsection}

		\begin{subsection}{Receiving Modes}
			\begin{figure}[!t]
				\centering
				\subfloat[TS receiver\label{fi:ts_receiver}]{
					\resizebox{0.9\columnwidth}{!}{
						\includegraphics{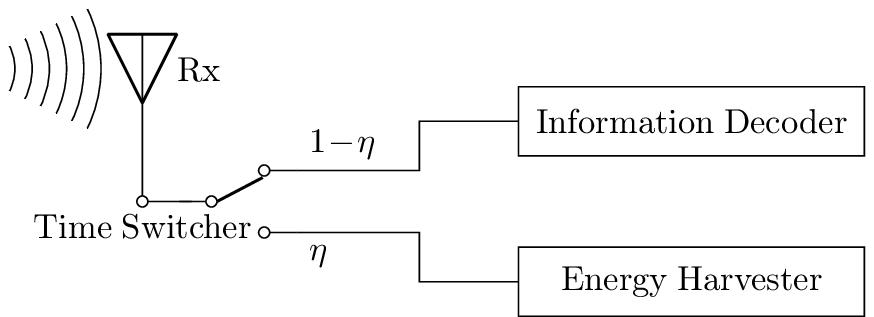}
					}
				}
				\\
				\subfloat[PS receiver\label{fi:ps_receiver}]{
					\resizebox{0.9\columnwidth}{!}{
						\includegraphics{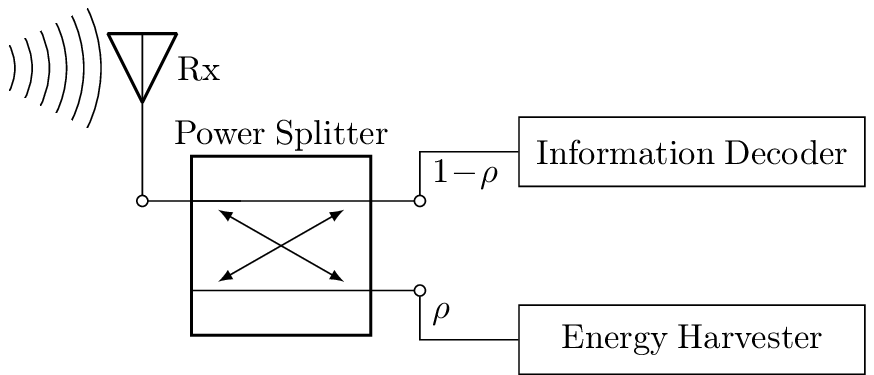}
					}
				}
				\caption{Diagrams of practical co-located receivers.}
				\label{fi:receiver}
			\end{figure}

			As illustrated in Fig.~\ref{fi:receiver}, there are two practical receiving modes for the co-located information decoder and energy harvester \cite{Zhou2013}. The TS receiver divides each transmission block into orthogonal data and energy sessions with duration $1-\eta$ and $\eta$, respectively. During each session, the transmitter optimizes the waveform for either Wireless Information Transfer (WIT) or WPT, while the receiver activates the information decoder or the energy harvester correspondingly. The duration ratio $\eta$ controls the R-E tradeoff and is independent from the waveform and beamforming design. On the other hand, the PS receiver splits the incoming signal into individual data and energy streams with power ratio $1-\rho$ and $\rho$, respectively. The data stream is fed into the information decoder while the energy stream is fed into the energy harvester. During each transmission block, the superposed waveform and splitting ratio are jointly designed to achieve different R-E tradeoffs. In the following context, we consider the optimization with the PS receiver, as the TS receiver can be regarded as a special case (i.e., a time sharing between $\rho=0$ and $\rho=1$).
		\end{subsection}

		\begin{subsection}{Information Decoder}
			A major benefit of the superposed waveform is that the multisine is deterministic and creates no interference to the modulated waveform \cite{Clerckx2018b}. Therefore, the achievable rate is\footnote{It requires waveform cancellation or translated demodulation \cite{Clerckx2018b}.}
			\begin{equation}\label{eq:R}
				R(\boldsymbol{\phi},\boldsymbol{w}_{\mathrm{I}},\rho) = \sum_{n=1}^N{\log_2\left(1+\frac{(1-\rho)\lvert \boldsymbol{h}_{n}^H\boldsymbol{w}_{\mathrm{I},n} \rvert^2}{\sigma_n^2}\right)},
			\end{equation}
			where $\sigma_n^2$ is the noise variance at the RF-band and during the RF-to-baseband conversion on tone $n$.
		\end{subsection}

		\begin{subsection}{Energy Harvester}\label{se:energy_harvester}
			\begin{figure}[!t]
				\centering
				\subfloat[Antenna equivalent circuit\label{fi:antenna_equivalent_circuit}]{
					\resizebox{0.45\columnwidth}{!}{
						\includegraphics{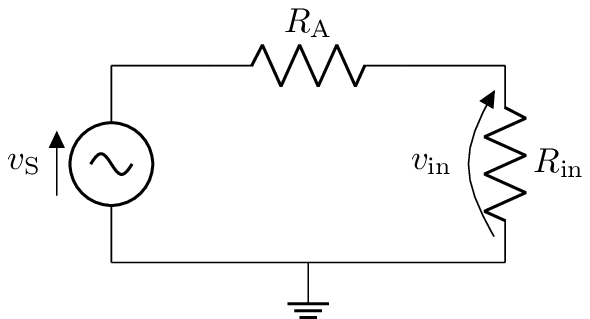}
					}
				}
				\subfloat[A single diode rectifier\label{fi:single_diode_rectifier}]{
					\resizebox{0.45\columnwidth}{!}{
						\includegraphics{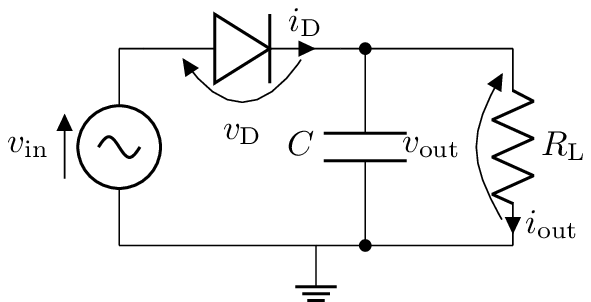}
					}
				}
				\caption{Equivalent circuits of receive antenna and energy harvester.}
			\end{figure}

			Taken from \cite{Clerckx2016a}, the rectenna model used in this section captures the dependency of the output DC on both the power and shape of the received signal. Fig.~\subref*{fi:antenna_equivalent_circuit} illustrates the equivalent circuit of an ideal antenna, where the antenna has a resistance $R_{\mathrm{A}}$ and the incoming signal creates a voltage source $v_{\mathrm{S}}(t)$. Let $R_{\mathrm{in}}$ be the total input resistance of the rectifier and matching network, and we assume the voltage across the matching network is negligible. When perfectly matched (i.e., $R_{\mathrm{in}}=R_{\mathrm{A}}$), the rectifier input voltage is $v_{\mathrm{in}}(t)=y(t)\sqrt{\rho R_{\mathrm{A}}}$. Consider a simplified rectifier model in Fig.~\subref*{fi:single_diode_rectifier} where a single series diode is followed by a low-pass filter with a parallel load. As detailed in \cite{Clerckx2018b}, a truncated Taylor expansion of the diode I-V characteristic equation suggests that maximizing the average output DC is equivalent to maximizing a monotonic function\footnote{This small-signal expansion model is only valid for the nonlinear operation region of the diode, and the I-V relationship would be linear if the diode behavior is dominated by the load \cite{Clerckx2016a}.}
			\begin{equation}\label{eq:z}
				z(\boldsymbol{\phi},\boldsymbol{w}_{\mathrm{I}},\boldsymbol{w}_{\mathrm{P}},\rho)=\sum_{i\,\mathrm{even},i\ge2}^{n_0}{k_i}{\rho^{i/2}}{R_{\mathrm{A}}^{i/2}}{\mathbb{E}\left\{\mathbb{A}\left\{y(t)^i\right\}\right\}},
			\end{equation}
			where $n_0$ is the truncation order and $k_i \triangleq i_{\mathrm{S}}/i!(n'v_{\mathrm{T}})^i$ is the diode coefficient ($i_{\mathrm{S}}$ is the reverse bias saturation current, $n'$ is the diode ideality factor, $v_{\mathrm{T}}$ is the thermal voltage). With a slight abuse of notation, we refer to $z$ as the average output DC in this paper. It can be observed that the conventional linear harvester model, where the output DC power equals the sum of the power harvested on each frequency, is a special case of \eqref{eq:z} with $n_0=2$. However, due to the coupling effect among different frequencies, some high-order AC components compensate each other in frequency and further contribute to the output DC power. In other words, even-order terms with $i \ge 4$ account for the nonlinear diode behavior. For simplicity, we choose $n_0=4$ to investigate the fundamental rectifier nonlinearity, and define $\beta_2 \triangleq {k_2}{R_{\mathrm{A}}}$, $\beta_4 \triangleq {k_4}{R_{\mathrm{A}}^2}$ to rewrite $z$ by \eqref{eq:z_expand}. Note that $\mathbb{E}\left\{\lvert\tilde{x}_{\mathrm{I},n}\rvert^2\right\}=1$ but $\mathbb{E}\left\{\lvert\tilde{x}_{\mathrm{I},n}\rvert^4\right\}=2$ applies a modulation gain on the fourth-order DC terms. Let $\boldsymbol{W}_{\mathrm{I/P}} \triangleq \boldsymbol{w}_{\mathrm{I/P}}\boldsymbol{w}_{\mathrm{I/P}}^H \in \mathbb{C}^{MN \times MN}$. As illustrated by Fig.~\ref{fi:block_diagonal}, $\boldsymbol{W}_{\mathrm{I/P}}$ can be divided into $N \times N$ blocks of size $M \times M$, and we let $\boldsymbol{W}_{\mathrm{I/P},k}$ keep its block diagonal $k \in \{-N+1,\dots,N-1\}$ and set all other blocks to $\boldsymbol{0}$. Hence, the components of $z$ reduce to \eqref{eq:y_I2}--\eqref{eq:y_P4} \cite{Golub2013}.

			\begin{figure*}[!t]
				\begin{equation*}
					\boldsymbol{W}_{\mathrm{I/P}}=
					\begin{tikzpicture}[>=stealth,thick,baseline,every right delimiter/.append style={name=rd},]
						\matrix [matrix of math nodes,left delimiter=(,right delimiter=)] (m)
						{
							\boldsymbol{w}_{\mathrm{I/P},1}\boldsymbol{w}_{\mathrm{I/P},1}^H & \boldsymbol{w}_{\mathrm{I/P},1}\boldsymbol{w}_{\mathrm{I/P},2}^H & \dots & \boldsymbol{w}_{\mathrm{I/P},1}\boldsymbol{w}_{\mathrm{I/P},N}^H \\
							\boldsymbol{w}_{\mathrm{I/P},2}\boldsymbol{w}_{\mathrm{I/P},1}^H & \boldsymbol{w}_{\mathrm{I/P},2}\boldsymbol{w}_{\mathrm{I/P},2}^H & \ddots & \vdots \\
							\vdots & \ddots & \ddots & \boldsymbol{w}_{\mathrm{I/P},N-1}\boldsymbol{w}_{\mathrm{I/P},N}^H \\
							\boldsymbol{w}_{\mathrm{I/P},N}\boldsymbol{w}_{\mathrm{I/P},1}^H & \dots & \boldsymbol{w}_{\mathrm{I/P},N}\boldsymbol{w}_{\mathrm{I/P},N-1}^H & \boldsymbol{w}_{\mathrm{I/P},N}\boldsymbol{w}_{\mathrm{I/P},N}^H \\
						};
						\draw[dotted,thick] (m-4-1.north west) rectangle (m-4-1.south east);
						\draw[dotted,thick,fill=gray,opacity=0.125] (m-2-1.north west) rectangle (m-2-1.south east); \draw[dotted,thick,fill=gray,opacity=0.125] (m-3-2.north west) rectangle (m-3-2.south east); \draw[dotted,thick,fill=gray,opacity=0.125] (m-4-3.north west) rectangle (m-4-3.south east);
						\draw[dashed,thick,fill=gray,opacity=0.5] (m-1-1.north west) rectangle (m-1-1.south east); \draw[dashed,thick,fill=gray,opacity=0.5] (m-2-2.north west) rectangle (m-2-2.south east); \draw[dashed,thick,fill=gray,opacity=0.5] (m-3-3.north west) rectangle (m-3-3.south east); \draw[dashed,thick,fill=gray,opacity=0.5] (m-4-4.north west) rectangle (m-4-4.south east);
						\draw[solid,thick,fill=gray,opacity=0.25] (m-1-2.north west) rectangle (m-1-2.south east); \draw[solid,thick,fill=gray,opacity=0.25] (m-2-3.north west) rectangle (m-2-3.south east); \draw[solid,thick,fill=gray,opacity=0.25] (m-3-4.north west) rectangle (m-3-4.south east);
						\draw[solid,thick] (m-1-4.north west) rectangle (m-1-4.south east);
						\draw[<-] (m-4-3.south|-m.south) -- ++(0.5,-0.15) node[below]{$k=-1$};
						\draw[<-] (rd.east|-m.south) -- ++(0.5,-0.15) node[right]{$k=0$};
						\draw[<-] (rd.east|-m-3-4.east) -- ++(0.5,-0.15) node[right]{$k=1$};
					\end{tikzpicture}
				\end{equation*}
				\caption{$\boldsymbol{W}_{\mathrm{I/P}}$ consists of $N \times N$ blocks of size $M \times M$. $\boldsymbol{W}_{\mathrm{I/P},k}$ keeps the $k$-th block diagonal of $\boldsymbol{W}_{\mathrm{I/P}}$ and nulls all remaining blocks. Solid, dashed and dotted blocks correspond to $k>0$, $k=0$ and $k<0$, respectively. For $\boldsymbol{w}_{\mathrm{I/P},n_1}\boldsymbol{w}_{\mathrm{I/P},n_2}^H$, the $k$-th block diagonal satisfies $k=n_2-n_1$.}
				\label{fi:block_diagonal}
			\end{figure*}
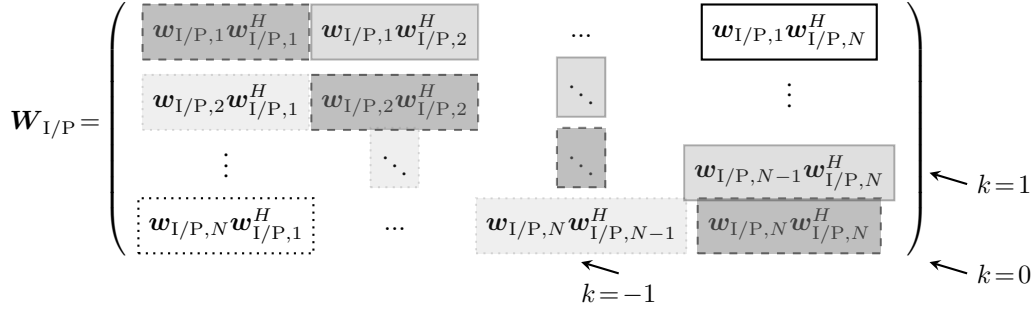

			\begin{figure*}[!b]
				\hrule
				\begin{align}
					&z(\boldsymbol{\phi},\boldsymbol{w}_{\mathrm{I}},\boldsymbol{w}_{\mathrm{P}},\rho) = \beta_2\rho\Bigl(\mathbb{E}\left\{\mathbb{A}\left\{y_{\mathrm{I}}^2(t)\right\}\right\}+\mathbb{A}\left\{y_{\mathrm{P}}^2(t)\right\}\Bigr)+\beta_4\rho^2\Bigl(\mathbb{E}\left\{\mathbb{A}\left\{y_{\mathrm{I}}^4(t)\right\}\right\}+\mathbb{A}\left\{y_{\mathrm{P}}^4(t)\right\}+6\mathbb{E}\left\{\mathbb{A}\left\{y_{\mathrm{I}}^2(t)\right\}\right\}\mathbb{A}\left\{y_{\mathrm{P}}^2(t)\right\}\Bigr),\label{eq:z_expand}\\
					&\mathbb{E}\left\{\mathbb{A}\left\{y_{\mathrm{I}}^2(t)\right\}\right\} = \frac{1}{2}\sum_{n=1}^N{(\boldsymbol{h}_{n}^H\boldsymbol{w}_{\mathrm{I},n})(\boldsymbol{h}_{n}^H\boldsymbol{w}_{\mathrm{I},n})^*} = \frac{1}{2}\boldsymbol{h}^H\boldsymbol{W}_{\mathrm{I},0}\boldsymbol{h},\label{eq:y_I2}\\
					&\mathbb{E}\left\{\mathbb{A}\left\{y_{\mathrm{I}}^4(t)\right\}\right\} = \frac{3}{4}\left(\sum_{n=1}^N{(\boldsymbol{h}_{n}^H\boldsymbol{w}_{\mathrm{I},n})(\boldsymbol{h}_{n}^H\boldsymbol{w}_{\mathrm{I},n})^*}\right)^2 = \frac{3}{4}(\boldsymbol{h}^H\boldsymbol{W}_{\mathrm{I},0}\boldsymbol{h})^2,\label{eq:y_I4}\\
					&\mathbb{A}\left\{y_{\mathrm{P}}^2(t)\right\} = \frac{1}{2}\sum_{n=1}^N{(\boldsymbol{h}_{n}^H\boldsymbol{w}_{\mathrm{P},n})(\boldsymbol{h}_{n}^H\boldsymbol{w}_{\mathrm{P},n})^*} = \frac{1}{2}\boldsymbol{h}^H\boldsymbol{W}_{\mathrm{P},0}\boldsymbol{h},\label{eq:y_P2}\\
					&\mathbb{A}\left\{y_{\mathrm{P}}^4(t)\right\} = \frac{3}{8}\sum_{\substack{{n_1},{n_2},{n_3},{n_4}\\{n_1}+{n_2}={n_3}+{n_4}}}{(\boldsymbol{h}_{{n_1}}^H\boldsymbol{w}_{\mathrm{P},{n_1}})(\boldsymbol{h}_{{n_2}}^H\boldsymbol{w}_{\mathrm{P},{n_2}})(\boldsymbol{h}_{{n_3}}^H\boldsymbol{w}_{\mathrm{P},{n_3}})^*(\boldsymbol{h}_{{n_4}}^H\boldsymbol{w}_{\mathrm{P},{n_4}})^*} = \frac{3}{8}\sum_{k=-N+1}^{N-1}(\boldsymbol{h}^H\boldsymbol{W}_{\mathrm{P},k}\boldsymbol{h})(\boldsymbol{h}^H\boldsymbol{W}_{\mathrm{P},k}\boldsymbol{h})^*.\label{eq:y_P4}
				\end{align}
			\end{figure*}
		\end{subsection}

		\begin{subsection}{Rate-Energy Region}
			The achievable R-E region is defined as
			\begin{align}
				\mathcal{C}_{\mathrm{R-E}}
				&\triangleq \biggl\{(R_{\mathrm{ID}}, z_{\mathrm{EH}}) \in \mathbb{R}_+^2 \mid R_{\mathrm{ID}} \le R, z_{\mathrm{EH}} \le z,\nonumber\\
				&\quad \frac{1}{2}\left(\lVert{\boldsymbol{w}_{\mathrm{I}}}\rVert^2+\lVert{\boldsymbol{w}_{\mathrm{P}}}\rVert^2\right) \le P\biggr\},
			\end{align}
			where $P$ is the average transmit power budget and \num{1/2} converts the peak value of the sine waves to the average value.
		\end{subsection}
	\end{section}

	\begin{section}{Problem Formulation}\label{se:problem_formulation}
		We characterize each R-E boundary point through a current maximization problem subject to sum rate, transmit power, and reflection amplitude constraints as
		\begin{maxi!}
			{\scriptstyle{\boldsymbol{\phi},\boldsymbol{w}_{\mathrm{I}},\boldsymbol{w}_{\mathrm{P}},\rho}}{z(\boldsymbol{\phi},\boldsymbol{w}_{\mathrm{I}},\boldsymbol{w}_{\mathrm{P}},\rho)}{\label{op:original}}{\label{ob:original}}
			\addConstraint{R(\boldsymbol{\phi},\boldsymbol{w}_{\mathrm{I}},\rho) \ge \bar{R}}\label{co:original_rate}
			\addConstraint{\frac{1}{2}\left(\lVert{\boldsymbol{w}_{\mathrm{I}}}\rVert^2+\lVert{\boldsymbol{w}_{\mathrm{P}}}\rVert^2\right)\le{P}}\label{co:original_power}
			\addConstraint{\lvert{\boldsymbol{\phi}}\rvert=\boldsymbol{1}}\label{co:original_modulus}
			\addConstraint{0 \le \rho \le 1.}
		\end{maxi!}
		Problem~\eqref{op:original} is intricate because of the coupled variables in \eqref{ob:original}, \eqref{co:original_rate} and the non-convex constraint \eqref{co:original_modulus}. To obtain a feasible solution, we propose a BCD algorithm that iteratively updates (i) the IRS phase shift; (ii) the active precoder; (iii) the waveform amplitude and splitting ratio, until convergence.

		\begin{subsection}{Passive Beamforming}
			In this section, we optimize the IRS phase shift $\boldsymbol{\phi}$ for any given waveform $\boldsymbol{w}_{\mathrm{I/P}}$ and splitting ratio $\rho$. Note that
			\begin{align}
				\lvert \boldsymbol{h}_{n}^H\boldsymbol{w}_{\mathrm{I},n} \rvert^2
				& = \boldsymbol{w}_{\mathrm{I},n}^H\boldsymbol{h}_n\boldsymbol{h}_n^H\boldsymbol{w}_{\mathrm{I},n}\nonumber\\
				& = \boldsymbol{w}_{\mathrm{I},n}^H(\boldsymbol{h}_{\mathrm{D},n}+\boldsymbol{V}_n^H\boldsymbol{\phi})(\boldsymbol{h}_{\mathrm{D},n}^H+\boldsymbol{\phi}^H\boldsymbol{V}_n)\boldsymbol{w}_{\mathrm{I},n}\nonumber\\
				& = \boldsymbol{w}_{\mathrm{I},n}^H\boldsymbol{M}_n^H\boldsymbol{\Phi}\boldsymbol{M}_n\boldsymbol{w}_{\mathrm{I},n}\nonumber\\
				& = \mathrm{tr}(\boldsymbol{M}_n\boldsymbol{w}_{\mathrm{I},n}\boldsymbol{w}_{\mathrm{I},n}^H\boldsymbol{M}_n^H\boldsymbol{\Phi})\nonumber\\
				& = \mathrm{tr}(\boldsymbol{C}_n\boldsymbol{\Phi}),
			\end{align}
			where $\boldsymbol{M}_n \triangleq [\boldsymbol{V}_n^H, \boldsymbol{h}_{\mathrm{D},n}]^H \in \mathbb{C}^{(L+1) \times M}$, $t'$ is an auxiliary variable with unit modulus, $\bar{\boldsymbol{\phi}} \triangleq [\boldsymbol{\phi}^H, t']^H \in \mathbb{C}^{(L+1) \times 1}$, $\boldsymbol{\Phi} \triangleq \bar{\boldsymbol{\phi}}\bar{\boldsymbol{\phi}}^H \in \mathbb{C}^{(L+1) \times (L+1)}$, $\boldsymbol{C}_n \triangleq \boldsymbol{M}_n\boldsymbol{w}_{\mathrm{I},n}\boldsymbol{w}_{\mathrm{I},n}^H\boldsymbol{M}_n^H \in \mathbb{C}^{(L+1)\times(L+1)}$. On the other hand, we define $t_{\mathrm{I/P},k}$ as
			\begin{align}
				t_{\mathrm{I/P},k}
				& \triangleq \boldsymbol{h}^H\boldsymbol{W}_{\mathrm{I/P},k}\boldsymbol{h}\nonumber\\
				& = \mathrm{tr}(\boldsymbol{h}\boldsymbol{h}^H\boldsymbol{W}_{\mathrm{I/P},k})\nonumber\\
				& = \mathrm{tr}\left((\boldsymbol{h}_{D}+\boldsymbol{V}^H\boldsymbol{\phi})(\boldsymbol{h}_{D}^H+\boldsymbol{\phi}^H\boldsymbol{V})\boldsymbol{W}_{\mathrm{I/P},k}\right)\nonumber\\
				& = \mathrm{tr}(\boldsymbol{M}^H\boldsymbol{\Phi}\boldsymbol{M}\boldsymbol{W}_{\mathrm{I/P},k})\nonumber\\
				& = \mathrm{tr}(\boldsymbol{M}\boldsymbol{W}_{\mathrm{I/P},k}\boldsymbol{M}^H\boldsymbol{\Phi})\nonumber\\
				& = \mathrm{tr}(\boldsymbol{C}_{\mathrm{I/P},k}\boldsymbol{\Phi})\label{eq:t_k},
			\end{align}
			where $\boldsymbol{V} \triangleq [\boldsymbol{V}_1,\dots,\boldsymbol{V}_N] \in \mathbb{C}^{L \times MN}$, $\boldsymbol{M} \triangleq [\boldsymbol{V}^H, \boldsymbol{h}_{D}]^H \in \mathbb{C}^{(L+1) \times MN}$, $\boldsymbol{C}_{\mathrm{I/P},k} \triangleq \boldsymbol{M}\boldsymbol{W}_{\mathrm{I/P},k}\boldsymbol{M}^H \in \mathbb{C}^{(L+1)\times(L+1)}$. On top of this, \eqref{eq:R} and \eqref{eq:z_expand} reduce respectively to
			\begin{align}
				R(\boldsymbol{\Phi})
				& = \sum_{n=1}^{N}{\log_2\left(1+\frac{(1-\rho)\mathrm{tr}(\boldsymbol{C}_n\boldsymbol{\Phi})}{\sigma_n^2}\right)},\label{eq:R_irs}\\
				z(\boldsymbol{\Phi})
				& = \frac{1}{2}{\beta_2}{\rho}(t_{\mathrm{I},0}+t_{\mathrm{P},0}) + \frac{3}{8}{\beta_4}{\rho^2} \left(2t_{\mathrm{I},0}^2 + \sum_{k=-N+1}^{N-1}{t_{\mathrm{P},k}t_{\mathrm{P},k}^*}\right)\nonumber\\
				& \quad + \frac{3}{2}{\beta_4}{\rho^2}t_{\mathrm{I},0}t_{\mathrm{P},0}.\label{eq:z_irs}
			\end{align}
			To maximize the non-concave expression \eqref{eq:z_irs}, we successively lower bound the second-order terms by their first-order Taylor expansions \cite{Adali2010}. Based on the solution at iteration $i - 1$, the approximations at iteration $i$ are
			\begin{align}
				(t_{\mathrm{I},0}^{(i)})^2
				& \ge 2 t_{\mathrm{I},0}^{(i)}t_{\mathrm{I},0}^{(i-1)} - (t_{\mathrm{I},0}^{(i-1)})^2,\label{eq:taylor_1}\\
				t_{\mathrm{P},k}^{(i)} (t_{\mathrm{P},k}^{(i)})^*
				& \ge 2 \Re\left\{t_{\mathrm{P},k}^{(i)} (t_{\mathrm{P},k}^{(i-1)})^*\right\} - t_{\mathrm{P},k}^{(i-1)} (t_{\mathrm{P},k}^{(i-1)})^*,\label{eq:taylor_2}\\
				t_{\mathrm{I},0}^{(i)} t_{\mathrm{P},0}^{(i)}
				& \ge t_{\mathrm{I},0}^{(i)} t_{\mathrm{P},0}^{(i-1)} + t_{\mathrm{P},0}^{(i)} t_{\mathrm{I},0}^{(i-1)} - t_{\mathrm{I},0}^{(i-1)} t_{\mathrm{P},0}^{(i-1)}.\label{eq:taylor_3}
			\end{align}
			Note that $t_{\mathrm{I/P},0}=\mathrm{tr}(\boldsymbol{C}_{\mathrm{I/P},0}\boldsymbol{\Phi})$ is real-valued because $\boldsymbol{C}_{\mathrm{I/P},0}$ and $\boldsymbol{\Phi}$ are Hermitian matrices. Due to symmetry \cite{Golub2013}, we have
			\begin{equation}\label{eq:coupled_terms}
				\sum_{k=-N+1}^{N-1} \Re\left\{t_{\mathrm{P},k}^{(i)} (t_{\mathrm{P},k}^{(i-1)})^*\right\} = \sum_{k=-N+1}^{N-1} t_{\mathrm{P},k}^{(i)} (t_{\mathrm{P},k}^{(i-1)})^*.
			\end{equation}
			Plugging \eqref{eq:taylor_1}--\eqref{eq:coupled_terms} into \eqref{eq:z_irs}, we obtain the DC approximation $\tilde{z}$ as \eqref{eq:z_irs_approx} and transform problem~\eqref{op:original} to
			\begin{figure*}[!b]
				\hrule
				\begin{align}
					\tilde{z}(\boldsymbol{\Phi}^{(i)})
					& = \frac{1}{2}{\beta_2}{\rho}(t_{\mathrm{I},0}^{(i)}+t_{\mathrm{P},0}^{(i)}) + \frac{3}{8}{\beta_4}{\rho^2} \left(4 t_{\mathrm{I},0}^{(i)}t_{\mathrm{I},0}^{(i-1)} - 2 (t_{\mathrm{I},0}^{(i-1)})^2 + \sum_{k=-N+1}^{N-1}{2 t_{\mathrm{P},k}^{(i)} (t_{\mathrm{P},k}^{(i-1)})^* - t_{\mathrm{P},k}^{(i-1)} (t_{\mathrm{P},k}^{(i-1)})^*}\right)\nonumber\\
					& \quad + \frac{3}{2}{\beta_4}{\rho^2} \left(t_{\mathrm{I},0}^{(i)} t_{\mathrm{P},0}^{(i-1)} + t_{\mathrm{P},0}^{(i)} t_{\mathrm{I},0}^{(i-1)} - t_{\mathrm{I},0}^{(i-1)} t_{\mathrm{P},0}^{(i-1)}\right).\label{eq:z_irs_approx}
				\end{align}
			\end{figure*}
			\begin{maxi!}
				{\scriptstyle{\boldsymbol{\Phi}}}{\tilde{z}(\boldsymbol{\Phi})}{\label{op:irs}}{\label{ob:irs}}
				\addConstraint{R(\boldsymbol{\Phi}) \ge \bar{R}}\label{co:irs_rate}
				\addConstraint{\mathrm{diag}^{-1}(\boldsymbol{\Phi})=\boldsymbol{1}}\label{co:irs_modulus}
				\addConstraint{\boldsymbol{\Phi}\succeq{\boldsymbol{0}}}\label{co:irs_sd}
				\addConstraint{\mathrm{rank}(\boldsymbol{\Phi})=1.\label{co:irs_rank}}
			\end{maxi!}
			We then apply Semi-Definite Relaxation (SDR) to the unit-rank constraint \eqref{co:irs_rank} and formulate a Semi-Definite Programming (SDP) with approximation accuracy no greater than $\pi/4$ \cite{Luo2010b}. In this specific case, we found the solution provided by CVX toolbox \cite{Grant2016} to \eqref{ob:irs}-\eqref{co:irs_sd} is always rank-\num{1}. This conclusion is summarized below.

			\begin{proposition}\label{pr:relaxation}
				Any optimal solution $\boldsymbol{\Phi}^\star$ to the relaxed passive beamforming problem~\eqref{ob:irs}--\eqref{co:irs_sd} is rank-\num{1} such that \eqref{co:irs_rank} is tight and no loss is introduced by SDR.
			\end{proposition}

			\begin{proof}\label{pf:relaxation}
				Please refer to Appendix~\ref{ap:relaxation}.
			\end{proof}

			In summary, we update $\boldsymbol{\Phi}^{(i)}$ until convergence, extract $\hat{\boldsymbol{\phi}}^\star$ by eigen decomposition, and retrieve the	IRS vector by $\boldsymbol{\phi}^{\star}=e^{j \arg\left([\hat{\boldsymbol{\phi}}^\star]_{(1:L)} \middle/ [\hat{\boldsymbol{\phi}}^\star]_{(L+1)}\right)}$. The passive beamforming design is summarized in the SCA Algorithm~\ref{al:sca}, where the relaxed problem \eqref{ob:irs}--\eqref{co:irs_sd} involves a $(L+1)$-order positive semi-definite matrix variable and $(L+2)$ linear constraints. Given a solution accuracy $\epsilon_{\mathrm{IPM}}$ for the interior-point method, the computational complexity of Algorithm~\ref{al:sca} is $\mathcal{O}\left(I_{\mathrm{SCA}}(L+2)^4 (L+1)^{0.5} \log(\epsilon_{\mathrm{IPM}}^{-1})\right)$, where $I_{\mathrm{SCA}}$ denotes the number of SCA iterations \cite{Luo2010b}.

			\begin{algorithm}[!t]
				\caption{SCA: IRS Phase Shift.}
				\label{al:sca}
				\begin{algorithmic}[1]
					\State \textbf{Input} $\beta_2$, $\beta_4$, $\boldsymbol{h}_{\mathrm{D},n}$, $\boldsymbol{V}_{n}$, $\sigma_n$, $\boldsymbol{w}_{\mathrm{I/P},n}$, $\rho$, $\bar{R}$, $\epsilon$, $\forall n$
					\State Construct $\boldsymbol{V}$, $\boldsymbol{M}$, $\boldsymbol{M}_n$, $\boldsymbol{C}_{n}$, $\boldsymbol{C}_{\mathrm{I/P},k}$, $\forall n,k$
					\State \textbf{Initialize} $i \gets 0$, $\boldsymbol{\Phi}^{(0)}$
					\State Set $t_{\mathrm{I/P},k}^{(0)}$, $\forall k$ by \eqref{eq:t_k}
					\State Compute $z^{(0)}$ by \eqref{eq:z_irs}
					\Repeat
						\State $i \gets i + 1$
						\State Get $\boldsymbol{\Phi}^{(i)}$ by solving \eqref{ob:irs}--\eqref{co:irs_sd}
						\State Update $t_{\mathrm{I/P},k}^{(i)}$, $\forall k$ by \eqref{eq:t_k}
						\State Compute $z^{(i)}$ by \eqref{eq:z_irs}
					\Until $\lvert z^{(i)}-z^{(i-1)} \rvert \le \epsilon$
					\State Set $\boldsymbol{\Phi}^{\star} \gets \boldsymbol{\Phi}^{(i)}$
					\State Get $\hat{\boldsymbol{\phi}}^\star$ by eigen decomposition, $\boldsymbol{\Phi}^{\star}=\hat{\boldsymbol{\phi}}^\star(\hat{\boldsymbol{\phi}}^\star)^H$
					\State Set $\boldsymbol{\phi}^{\star} \gets e^{j \arg\left([\hat{\boldsymbol{\phi}}^\star]_{(1:L)} \middle/ [\hat{\boldsymbol{\phi}}^\star]_{(L+1)}\right)}$
					\State \textbf{Output} $\boldsymbol{\phi}^{\star}$
				\end{algorithmic}
			\end{algorithm}

			\begin{proposition}\label{pr:sca}
				For any feasible initial point with given waveform and splitting ratio, the SCA Algorithm~\ref{al:sca} is guaranteed to converge to local optimal points of the original problem~\eqref{op:original}.
			\end{proposition}

			\begin{proof}\label{pf:sca}
				Please refer to Appendix~\ref{ap:sca}.
			\end{proof}
		\end{subsection}

		\begin{subsection}{Active Beamforming}
			The original waveform and active beamforming problem~\eqref{op:original} is over complex vectors $\boldsymbol{w}_{\mathrm{I/P}}$ of size $MN \times 1$. Next, we decouple the design in spatial and frequency domains, enable independent optimizations correspondingly, and reduce the size of variables from $2MN$ to $2(M+N)$. The weight on subband $n$ is essentially
			\begin{equation}\label{eq:w}
				\boldsymbol{w}_{\mathrm{I/P}, n} = s_{\mathrm{I/P}, n} \boldsymbol{b}_{\mathrm{I/P}, n},
			\end{equation}
			where $s_{\mathrm{I/P},n}$ denotes the amplitude of the modulated/multisine waveform at tone $n$, and $\boldsymbol{b}_{\mathrm{I/P}, n}$ denotes the corresponding information/power precoder. Define $\boldsymbol{s}_{\mathrm{I/P}} \triangleq [s_{\mathrm{I/P},1},\dots,s_{\mathrm{I/P},N}]^T \in \mathbb{R}_+^{N \times 1}$. The MRT precoder at subband $n$ is given by
			\begin{equation}\label{eq:b_n}
				\boldsymbol{b}_{\mathrm{I/P}, n}^\star = \frac{\boldsymbol{h}_n}{\lVert{\boldsymbol{h}_n}\rVert}.
			\end{equation}

			\begin{proposition}\label{pr:mrt}
				For single-user SWIPT, the global optimal information and power precoders coincide at the MRT.
			\end{proposition}

			\begin{proof}\label{pf:mrt}
				Please refer to Appendix~\ref{ap:mrt}.
			\end{proof}
		\end{subsection}

		\begin{subsection}{Waveform and Splitting Ratio}
			Next, we jointly optimize the waveform amplitude $\boldsymbol{s}_{\mathrm{I/P}}$ and the splitting ratio $\rho$ for any given IRS phase shift $\boldsymbol{\phi}$ and active precoder $\boldsymbol{b}_{\mathrm{I/P},n}$, $\forall n$. On top of \eqref{eq:b_n}, the equivalent channel strength at subband $n$ is $\lVert{\boldsymbol{h}_n}\rVert$. Hence, the rate \eqref{eq:R} reduces to
			\begin{equation}\label{eq:R_waveform}
				R(\boldsymbol{s}_{\mathrm{I}},\rho) = \log_2\prod_{n=1}^N\left(1+\frac{(1-\rho)\lVert{\boldsymbol{h}_n}\rVert^2 s_{\mathrm{I},n}^2}{\sigma_n^2}\right),
			\end{equation}
			and the DC \eqref{eq:z_expand} rewrites as \eqref{eq:z_waveform}, so that problem~\eqref{op:original} boils down to
			\begin{figure*}[!b]
				\hrule
				\begin{align}
					z(\boldsymbol{s}_{\mathrm{I}},\boldsymbol{s}_\mathrm{P},\rho)
					& = \frac{1}{2}{\beta_2}{\rho} \sum_{n=1}^N \lVert{\boldsymbol{h}_n}\rVert^2(s_{\mathrm{I},n}^2+s_{\mathrm{P},n}^2) + \frac{3}{8}{\beta_4}{\rho^2} \left( 2\sum_{n_1,n_2} \prod_{j=1}^2 \lVert{\boldsymbol{h}_{n_j}}\rVert^2 s_{\mathrm{I},{n_j}}^2 + \sum_{\substack{{n_1},{n_2},{n_3},{n_4}\\{n_1}+{n_2}={n_3}+{n_4}}} \prod_{j=1}^4 \lVert{\boldsymbol{h}_{n_j}}\rVert s_{\mathrm{P},{n_j}} \right)\nonumber\\
					& \quad + \frac{3}{2}{\beta_4}{\rho^2} \left( \sum_{n_1,n_2} \lVert{\boldsymbol{h}_{n_1}}\rVert^2 \lVert{\boldsymbol{h}_{n_2}}\rVert^2 s_{\mathrm{I},{n_1}}^2 s_{\mathrm{P},{n_2}}^2 \right).\label{eq:z_waveform}
				\end{align}
			\end{figure*}
			\begin{maxi!}
				{\scriptstyle{\boldsymbol{s}_{\mathrm{I}},\boldsymbol{s}_\mathrm{P},\rho}}{z(\boldsymbol{s}_{\mathrm{I}},\boldsymbol{s}_\mathrm{P},\rho)}{\label{op:waveform}}{}
				\addConstraint{R(\boldsymbol{s}_{\mathrm{I}},\rho) \ge \bar{R}}
				\addConstraint{\frac{1}{2}\left(\lVert{\boldsymbol{s}_{\mathrm{I}}}\rVert^2+\lVert{\boldsymbol{s}_\mathrm{P}}\rVert^2\right)\le{P}.}
			\end{maxi!}
			Following \cite{Clerckx2018b}, we introduce auxiliary variables $t'',\bar{\rho}$ and transform problem~\eqref{op:waveform} into a reversed GP
			\begin{mini!}
				{\scriptstyle{\boldsymbol{s}_{\mathrm{I}},\boldsymbol{s}_\mathrm{P},\rho,\bar{\rho},t''}}{\frac{1}{t''}}{\label{op:waveform_rgp}}{}
				\addConstraint{\frac{t''}{z(\boldsymbol{s}_{\mathrm{I}},\boldsymbol{s}_\mathrm{P},\rho)} \le 1}\label{co:waveform_objective}
				\addConstraint{\frac{2^{\bar{R}}}{\prod_{n=1}^N \left(1+{\bar{\rho}\lVert{\boldsymbol{h}_n}\rVert^2 s_{\mathrm{I},n}^2}\big/{\sigma_n^2}\right)} \le 1}\label{co:waveform_rate}
				\addConstraint{\frac{1}{2}\left(\lVert{\boldsymbol{s}_{\mathrm{I}}}\rVert^2+\lVert{\boldsymbol{s}_\mathrm{P}}\rVert^2\right) \le P}\label{co:waveform_power}
				\addConstraint{\rho + \bar{\rho} \le 1.}\label{co:waveform_splitting_ratio}
			\end{mini!}
			It can be concluded that $\bar{\rho}^{\star}=1-\rho^{\star}$ as no power is wasted at the receiver. The denominators of \eqref{co:waveform_rate} and \eqref{co:waveform_objective} consist of posynomials \cite{Boyd2007} that can be decomposed as sums of monomials
			\begin{align}
				1+\frac{\bar{\rho}\lVert{\boldsymbol{h}_n}\rVert^2 s_{\mathrm{I},n}^2}{\sigma_n^2} &\triangleq \sum_{m_{\mathrm{I},n}}g_{m_{\mathrm{I},n}}(s_{\mathrm{I},n},\bar{\rho})\label{eq:g_I},\\
				z(\boldsymbol{s}_{\mathrm{I}},\boldsymbol{s}_\mathrm{P},\rho) &\triangleq \sum_{m_\mathrm{P}}{g_{m_\mathrm{P}}(\boldsymbol{s}_{\mathrm{I}},\boldsymbol{s}_\mathrm{P},\rho)}\label{eq:g_P}.
			\end{align}
			We upper bound \eqref{eq:g_I} and \eqref{eq:g_P} by the Arithmetic Mean-Geometric Mean (AM-GM) inequality \cite{Chiang2005} and transform problem~\eqref{op:waveform_rgp} to
			\begin{mini!}
				{\scriptstyle{\boldsymbol{s}_{\mathrm{I}},\boldsymbol{s}_\mathrm{P},\rho,\bar{\rho},t''}}{\frac{1}{t''}}{\label{op:waveform_gp}}{}
				\addConstraint{{t''}\prod_{m_\mathrm{P}}{\left(\frac{g_{{m_\mathrm{P}}}(\boldsymbol{s}_{\mathrm{I}},\boldsymbol{s}_\mathrm{P},\rho)}{\gamma_{{m_\mathrm{P}}}}\right)^{-\gamma_{{m_\mathrm{P}}}}}\le{1}}
				\addConstraint{2^{\bar{R}}\prod_{n}\prod_{m_{\mathrm{I},n}}\left(\frac{g_{m_{\mathrm{I},n}}(s_{\mathrm{I},n},\bar{\rho})}{\gamma_{m_{\mathrm{I},n}}}\right)^{-\gamma_{m_{\mathrm{I},n}}}\le{1}}
				\addConstraint{\frac{1}{2}\left(\lVert{\boldsymbol{s}_{\mathrm{I}}}\rVert^2+\lVert{\boldsymbol{s}_\mathrm{P}}\rVert^2\right)\le{P}}
				\addConstraint{\rho + \bar{\rho} \le 1,}
			\end{mini!}
			where $\gamma_{m_{\mathrm{I},n}},\gamma_{m_\mathrm{P}} \ge 0$ and $\sum_{m_{\mathrm{I},n}}\gamma_{m_{\mathrm{I},n}}=\sum_{m_\mathrm{P}}\gamma_{m_\mathrm{P}}=1$. The tightness of the AM-GM inequality depends on $\{\gamma_{m_{\mathrm{I},n}},\gamma_{m_\mathrm{P}}\}$, and a feasible choice at iteration $i$ is
			\begin{align}
				\gamma_{m_{\mathrm{I},n}}^{(i)} & = \frac{g_{m_{\mathrm{I},n}}(s_{\mathrm{I},n}^{(i-1)},\bar{\rho}^{(i-1)})}{1+{\bar{\rho}^{(i-1)}\lVert{\boldsymbol{h}_n}\rVert^2 (s_{\mathrm{I},n}^{(i-1)})^2}\big/{\sigma_n^2}}\label{eq:gamma_I},\\
				\gamma_{m_\mathrm{P}}^{(i)} & = \frac{g_{m_\mathrm{P}}(\boldsymbol{s}_{\mathrm{I}}^{(i-1)},\boldsymbol{s}_\mathrm{P}^{(i-1)},\rho^{(i-1)})}{z(\boldsymbol{s}_{\mathrm{I}}^{(i-1)},\boldsymbol{s}_\mathrm{P}^{(i-1)},\rho^{(i-1)})}\label{eq:gamma_P}.
			\end{align}
			With \eqref{eq:gamma_I} and \eqref{eq:gamma_P}, problem~\eqref{op:waveform_gp} can be solved by existing optimization tools such as CVX \cite{Grant2016}. We update $\boldsymbol{s}_{\mathrm{I}}^{(i)},\boldsymbol{s}_\mathrm{P}^{(i)},\rho^{(i)}$ iteratively until convergence. The joint waveform amplitude and splitting ratio design is summarized in the GP Algorithm~\ref{al:gp}, which achieves local optimality at the cost of exponential computational complexity \cite{Chiang2005}.

			\begin{algorithm}[!t]
				\caption{GP: Waveform Amplitude and Splitting Ratio.}
				\label{al:gp}
				\begin{algorithmic}[1]
					\State \textbf{Input} $\beta_2$, $\beta_4$, $\boldsymbol{h}_n$, $P$, $\sigma_n$, $\bar{R}$, $\epsilon$, $\forall n$
					\State \textbf{Initialize} $i \gets 0$, $\boldsymbol{s}_{\mathrm{I/P}}^{(0)}$, $\rho^{(0)}$
					\State Compute $R^{(0)}$, $z^{(0)}$ by \eqref{eq:R_waveform}, \eqref{eq:z_waveform}
					\State Set $g_{m_{\mathrm{I},n}}^{(0)}$, $g_{m_\mathrm{P}}^{(0)}$, $\forall n$ by \eqref{eq:g_I}, \eqref{eq:g_P}
					\Repeat
						\State $i \gets i + 1$
						\State Update $\gamma_{m_{\mathrm{I},n}}^{(i)}$, $\gamma_{m_\mathrm{P}}^{(i)}$, $\forall n$ by \eqref{eq:gamma_I}, \eqref{eq:gamma_P}
						\State Get $\boldsymbol{s}_{\mathrm{I/P}}^{(i)}$, $\rho^{(i)}$ by solving problem~\eqref{op:waveform_gp}
						\State Compute $R^{(i)}$, $z^{(i)}$ by \eqref{eq:R_waveform}, \eqref{eq:z_waveform}
						\State Update $g_{m_{\mathrm{I},n}}^{(i)}$, $g_{m_\mathrm{P}}^{(i)}$, $\forall n$ by \eqref{eq:g_I}, \eqref{eq:g_P}
					\Until $\lvert z^{(i)} - z^{(i-1)} \rvert \le \epsilon$
					\State Set $\boldsymbol{s}_{\mathrm{I/P}}^{\star} \gets \boldsymbol{s}_{\mathrm{I/P}}^{(i)}$, $\rho^{\star} \gets \rho^{(i)}$
					\State \textbf{Output} $\boldsymbol{s}_{\mathrm{I}}^{\star}$, $\boldsymbol{s}_{\mathrm{P}}^{\star}$, $\rho^{\star}$
				\end{algorithmic}
			\end{algorithm}

			\begin{proposition}\label{pr:gp}
				For any feasible initial point, the GP Algorithm~\ref{al:gp} is guaranteed to converge to local optimal points of the waveform amplitude and splitting ratio design problem \eqref{op:waveform}.
			\end{proposition}

			\begin{proof}\label{pf:gp}
				Please refer to \cite{Clerckx2016a,Clerckx2018b}.
			\end{proof}
		\end{subsection}

		\begin{subsection}{Low-Complexity Adaptive Design}
			To facilitate practical SWIPT implementation, we propose two closed-form adaptive waveform amplitude schemes by combining WF and SMF in time and power domains, respectively. For WIT, the optimal WF strategy assigns the amplitude of modulated tone $n$ by
			\begin{equation}\label{eq:wf}
				s_{\mathrm{I}, n} = \sqrt{2\left(\lambda - \frac{\sigma_n^2}{P \lVert{\boldsymbol{h}_n}\rVert^2}\right)^+},
			\end{equation}
			where $\lambda$ is chosen to satisfy the power constraint $\lVert{\boldsymbol{s}_I}\rVert^2 / 2 \le P$. The closed-form solution can be obtained by iterative power allocation \cite{Tse2005}, and the details are omitted here. On the other hand, SMF was proposed in \cite{Clerckx2017} as a suboptimal WPT resource allocation scheme that assigns the amplitude of sinewave $n$ by
			\begin{equation}\label{eq:smf}
				s_{\mathrm{P}, n} = \sqrt{\frac{2 P}{\sum_{n=1}^N \lVert{\boldsymbol{h}_n \rVert^{2 \alpha}}}}\lVert{\boldsymbol{h}_n}\rVert^\alpha,
			\end{equation}
			where the scaling ratio $\alpha \ge 1$ is predetermined to exploit the rectifier nonlinearity and frequency selectivity. When the receiver works in TS mode, there is no superposition in the suboptimal waveform design (modulated waveform with amplitude \eqref{eq:wf} is used in the data session while multisine waveform with amplitude \eqref{eq:smf} is used in the energy session). When the receiver works in PS mode, we jointly design the combining ratio $\delta$ with the splitting ratio $\rho$, and assign the superposed waveform amplitudes as
			\begin{align}
				s_{\mathrm{I}, n} &= \sqrt{2(1 - \delta)\left(\lambda - \frac{\sigma_n^2}{P \lVert{\boldsymbol{h}_n}\rVert^2}\right)^+}, \label{eq:s_i}\\
				s_{\mathrm{P}, n} &= \sqrt{\frac{2 \delta P}{\sum_{n=1}^N \lVert{\boldsymbol{h}_n \rVert^{2 \alpha}}}}\lVert{\boldsymbol{h}_n}\rVert^\alpha, \label{eq:s_p}
			\end{align}
			where the $\delta$ determines the power ratio of multisine waveform at the transmitter, and $\rho$ determines the power ratio of the energy harvester at the receiver.\footnote{We notice that $\delta^{\star}=\rho^{\star}=0$ at the WIT point and $\delta^{\star}=\rho^{\star}=1$ at the WPT point when $N$ is relatively large. Intuitively, $\delta^{\star}$ and $\rho^{\star}$ should be positively correlated for efficient SWIPT design.}

			Besides, minor modifications are required for passive beamforming to accommodate the low-complexity waveform schemes. Specifically, the rate constraint \eqref{co:irs_rate} should be dropped as the achievable rate is controlled by $\eta$ or $\{\delta,\rho\}$. To achieve the WIT point ($\rho=0$), the rate \eqref{eq:R_irs} should be maximized, the current expression \eqref{eq:z_irs_approx} is not needed and no SCA is involved. The Modified-SCA (M-SCA) Algorithm~\ref{al:m_sca} summarizes the modified passive beamforming design when the receiver works in PS mode. Similarly, no loss is introduced by SDR and local optimality is guaranteed. The proofs are omitted here. Since each SDP involves $(L+1)$ linear constraints, the computational complexity of Algorithm~\ref{al:m_sca} is $\mathcal{O}\left(I_{\mathrm{M-SCA}}(L+1)^{4.5} \log(\epsilon_{\mathrm{IPM}}^{-1})\right)$, where $I_{\mathrm{M-SCA}}$ denotes the number of M-SCA iterations \cite{Luo2010b}. Note that no SCA is involved at the WIT point where $I_{\mathrm{M-SCA}}=1$.

			\begin{algorithm}[!t]
				\caption{M-SCA: IRS Phase Shift.}
				\label{al:m_sca}
				\begin{algorithmic}[1]
					\State \textbf{Input} $\beta_2$, $\beta_4$, $\boldsymbol{h}_{\mathrm{D},n}$, $\boldsymbol{V}_{n}$, $\sigma_n$, $\boldsymbol{w}_{\mathrm{I/P},n}$, $\rho$, $\epsilon$, $\forall n$
					\State Construct $\boldsymbol{V}$, $\boldsymbol{M}$, $\boldsymbol{M}_n$, $\boldsymbol{C}_{n}$, $\boldsymbol{C}_{\mathrm{I/P},k}$, $\forall n,k$
					\State \textbf{Initialize} $i \gets 0$, $\boldsymbol{\Phi}^{(0)}$
					\If{$\rho=0$}
						\State Get $\boldsymbol{\Phi}^{\star}$ by maximizing \eqref{eq:R_irs} s.t. \eqref{co:irs_modulus}, \eqref{co:irs_sd}
					\Else
						\State Set $t_{\mathrm{I/P},k}^{(0)}$, $\forall k$ by \eqref{eq:t_k}
						\State Compute $z^{(0)}$ by \eqref{eq:z_irs}
						\Repeat
							\State $i \gets i + 1$
								\State Get $\boldsymbol{\Phi}^{(i)}$ by maximizing \eqref{eq:z_irs_approx} s.t. \eqref{co:irs_modulus}, \eqref{co:irs_sd}
								\State Update $t_{\mathrm{I/P},k}^{(i)}$, $\forall k$ by \eqref{eq:t_k}
								\State Compute $z^{(i)}$ by \eqref{eq:z_irs}
						\Until $\lvert z^{(i)}-z^{(i-1)} \rvert \le \epsilon$
						\State Set $\boldsymbol{\Phi}^{\star} \gets \boldsymbol{\Phi}^{(i)}$
					\EndIf
					\State Get $\hat{\boldsymbol{\phi}}^\star$ by eigen decomposition, $\boldsymbol{\Phi}^{\star}=\hat{\boldsymbol{\phi}}^\star(\hat{\boldsymbol{\phi}}^\star)^H$
					\State Set $\boldsymbol{\phi}^{\star} \gets e^{j \arg\left([\hat{\boldsymbol{\phi}}^\star]_{(1:L)} \middle/ [\hat{\boldsymbol{\phi}}^\star]_{(L+1)}\right)}$
					\State \textbf{Output} $\boldsymbol{\phi}^{\star}$
				\end{algorithmic}
			\end{algorithm}
		\end{subsection}

		\begin{subsection}{Block Coordinate Descent}
			Based on the direct and cascaded CSIT, we iteratively update the passive beamforming $\boldsymbol{\phi}$ by Algorithm~\ref{al:sca}, the active precoder $\boldsymbol{b}_{\mathrm{I/P},n}$, $\forall n$ by equation \eqref{eq:b_n}, and the waveform amplitude $\boldsymbol{s}_{\mathrm{I/P}}$ and splitting ratio $\rho$ by Algorithm~\ref{al:gp}, until convergence. The steps are summarized in the BCD Algorithm~\ref{al:bcd}, whose computational complexity is exponential as inherited from Algorithm~\ref{al:gp}. It is guaranteed to converge, but may end up with a suboptimal solution because variables are coupled in constraint~\eqref{co:original_rate} \cite{Grippo2000}.

			\begin{algorithm}[!t]
				\caption{BCD: Waveform, Beamforming and Splitting Ratio.}
				\label{al:bcd}
				\begin{algorithmic}[1]
					\State \textbf{Input} $\beta_2$, $\beta_4$, $\boldsymbol{h}_{\mathrm{D},n}$, $\boldsymbol{V}_{n}$, $P$, $\sigma_n$, $\bar{R}$, $\epsilon$, $\forall n$
					\State \textbf{Initialize} $i \gets 0$, $\boldsymbol{\phi}^{(0)}$, $\boldsymbol{b}_{\mathrm{I/P},n}^{(0)}$, $\boldsymbol{s}_{\mathrm{I/P}}^{(0)}$, $\rho^{(0)}$, $\forall n$
					\State Set $\boldsymbol{w}_{\mathrm{I/P},n}^{(0)}$, $\forall n$ by \eqref{eq:w}
					\State Compute $z^{(0)}$ by \eqref{eq:z_waveform}
					\Repeat
						\State $i \gets i + 1$
						\State Get $\boldsymbol{\phi}^{(i)}$ based on $\boldsymbol{w}_{\mathrm{I/P}}^{(i-1)}$, $\rho^{(i-1)}$ by Algorithm~\ref{al:sca}
						\State Update $\boldsymbol{h}_n^{(i)}$, $\boldsymbol{b}_n^{(i)}$, $\forall n$ by \eqref{eq:h_n}, \eqref{eq:b_n}
						\State Get $\boldsymbol{s}_{\mathrm{I/P}}^{(i)}$, $\rho^{(i)}$ by Algorithm~\ref{al:gp}
						\State Update $\boldsymbol{w}_{\mathrm{I/P},n}^{(i)}$, $\forall n$ by \eqref{eq:w}
						\State Compute $z^{(i)}$ by \eqref{eq:z_waveform}
					\Until $\lvert z^{(i)} - z^{(i-1)} \rvert \le \epsilon$
					\State Set $\boldsymbol{\phi}^{\star} \gets \boldsymbol{\phi}^{(i)}$, $\boldsymbol{w}_{\mathrm{I/P}}^{\star} \gets \boldsymbol{w}_{\mathrm{I/P}}^{(i)}$, $\rho^{\star} \gets \rho^{(i)}$
					\State \textbf{Output} $\boldsymbol{\phi}^{\star}$, $\boldsymbol{w}_{\mathrm{I}}^{\star}$, $\boldsymbol{w}_{\mathrm{P}}^{\star}$, $\rho^{\star}$
				\end{algorithmic}
			\end{algorithm}

			For the low-complexity design under PS mode, we obtain the phase shift by Algorithm~\ref{al:m_sca}, the active precoder $\boldsymbol{b}_{\mathrm{I/P},n}$, $\forall n$ by equation \eqref{eq:b_n}, and the waveform amplitude by \eqref{eq:s_i} and \eqref{eq:s_p}. To achieve the WIT point ($\rho=0$), the rate \eqref{eq:R_irs} should be maximized to obtain the maximum capacity $C_{\max}$.\footnote{Recall in Remark~\ref{re:subband_tradeoff} that different subchannel designs lead to different capacities.} Note that the BCD algorithm obtains the R-E region by varying the rate constraint from \num{0} to $C_{\max}$, while the achievable R-E region of the LC-BCD algorithm can be obtained by performing a two-dimensional search over $(\delta, \rho)$ from $(0, 0)$ to $(1, 1)$. The steps are summarized in the Low Complexity-BCD (LC-BCD) Algorithm~\ref{al:lc_bcd}. The computational complexity of Algorithm~\ref{al:lc_bcd} is $\mathcal{O}\left(I_{\mathrm{LC-BCD}}I_{\mathrm{M-SCA}}(L+1)^{4.5} \log(\epsilon_{\mathrm{IPM}}^{-1})\right)$, where $I_{\mathrm{LC-BCD}}$ denotes the number of LC-BCD iterations \cite{Luo2010b}.

			\begin{algorithm}[!t]
				\caption{LC-BCD: Waveform and Beamforming.}
				\label{al:lc_bcd}
				\begin{algorithmic}[1]
					\State \textbf{Input} $\beta_2$, $\beta_4$, $\boldsymbol{h}_{\mathrm{D},n}$, $\boldsymbol{V}_{n}$, $P$, $\sigma_n$, $\delta$, $\rho$, $\epsilon$, $\forall n$
					\State \textbf{Initialize} $i \gets 0$, $\boldsymbol{\phi}^{(0)}$, $\boldsymbol{b}_{\mathrm{I/P},n}^{(0)}$, $\boldsymbol{s}_{\mathrm{I/P}}^{(0)}$, $\forall n$
					\State Set $\boldsymbol{w}_{\mathrm{I/P},n}^{(0)}$, $\forall n$ by \eqref{eq:w}
					\State Compute $R^{(0)}$, $z^{(0)}$ by \eqref{eq:R_waveform}, \eqref{eq:z_waveform}
					\Repeat
						\State $i \gets i + 1$
						\State Get $\boldsymbol{\phi}^{(i)}$ based on $\boldsymbol{w}_{\mathrm{I/P}}^{(i-1)}$ by Algorithm~\ref{al:m_sca}
						\State Update $\boldsymbol{h}_n^{(i)}$, $\boldsymbol{b}_n^{(i)}$, $\forall n$ by \eqref{eq:h_n}, \eqref{eq:b_n}
						\State Update $\boldsymbol{s}_{\mathrm{I}}^{(i)}$, $\boldsymbol{s}_{\mathrm{P}}^{(i)}$ by \eqref{eq:s_i}, \eqref{eq:s_p}
						\State Update $\boldsymbol{w}_{\mathrm{I/P},n}^{(i)}$, $\forall n$ by \eqref{eq:w}
						\State Compute $R^{(i)}$, $z^{(i)}$ by \eqref{eq:R_waveform}, \eqref{eq:z_waveform}
						\If{$\rho=0$}
							\State $\Delta \gets R^{(i)} - R^{(i-1)}$
						\Else
							\State $\Delta \gets z^{(i)} - z^{(i-1)}$
						\EndIf
					\Until $\lvert \Delta \rvert \le \epsilon$
					\State Set $\boldsymbol{\phi}^{\star} \gets \boldsymbol{\phi}^{(i)}$, $\boldsymbol{w}_{\mathrm{I/P}}^{\star} \gets \boldsymbol{w}_{\mathrm{I/P}}^{(i)}$
					\State \textbf{Output} $\boldsymbol{\phi}^{\star}$, $\boldsymbol{w}_{\mathrm{I}}^{\star}$, $\boldsymbol{w}_{\mathrm{P}}^{\star}$
				\end{algorithmic}
			\end{algorithm}
		\end{subsection}
	\end{section}

	\begin{section}{Performance Evaluations}\label{se:performance_evaluation}
		\begin{figure}[!t]
			\centering
			\def\svgwidth{0.9\columnwidth}
			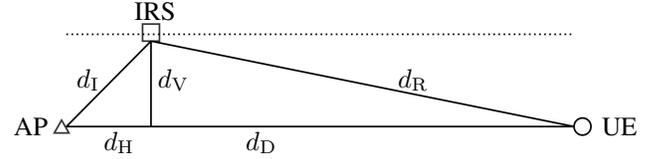
			\caption{System layout in simulation.}
			\label{fi:layout}
		\end{figure}

		To evaluate the proposed IRS-aided SWIPT system, we consider the layout in Fig.~\ref{fi:layout} where the IRS moves along a line parallel to the AP-UE path. Let $d_{\mathrm{H}}$, $d_{\mathrm{V}}$ be the horizontal and vertical distances from the AP to the IRS, and denote respectively $d_{\mathrm{D}}$, $d_{\mathrm{I}}=\sqrt{d_{\mathrm{H}}^2+d_{\mathrm{V}}^2}$, $d_{\mathrm{R}}=\sqrt{(d_{\mathrm{D}}-d_{\mathrm{H}})^2+d_{\mathrm{V}}^2}$ as the distance of direct, incident and reflected links. $d_{\mathrm{D}}=\SI{12}{\meter}$ and $d_{\mathrm{H}}=d_{\mathrm{V}}=\SI{2}{\meter}$ are chosen as reference. The path loss of direct, incident and reflected links are denoted by $\Lambda_{\mathrm{D}}$, $\Lambda_{\mathrm{I}}$ and $\Lambda_{\mathrm{R}}$, respectively. We consider a large open space Wi-Fi-like environment at center frequency \SI{2.4}{\GHz} where the channel follows IEEE TGn channel model D \cite{Erceg2004}. Specifically, the path loss exponent is \num{2} (i.e., free-space model) up to \SI{10}{\meter}, and \num{3.5} onwards to further penalize the channels with large distance. All fadings are modeled as Non-LoS (NLoS) with tap delays and powers specified in model D, and the tap gains are modeled as i.i.d. CSCG variables. Rectenna parameters are set to $k_2=0.0034$, $k_4=0.3829$, $R_{\mathrm{A}}=\SI{50}{\ohm}$ \cite{Clerckx2016a} such that $\beta_2=0.17$ and $\beta_4=957.25$. We also choose the average Effective Isotropic Radiated Power (EIRP) as $P=\SI{40}{\dBm}$, the receive antenna gain as \SI{3}{\dBi}, the scaling ratio as $\alpha=2$, and the tolerance as $\epsilon=10^{-8}$. To further reduce the complexity, we assume $\delta=\rho$ for simplicity and perform a one-dimensional search from \num{0} to \num{1} to obtain a inner R-E bound for the LC-BCD algorithm. Each R-E point is averaged over \num{200} channel realizations, and the $x$-axis is normalized to per-subband rate $R/N$.

		\begin{figure}[!t]
			\centering
			\resizebox{0.8\columnwidth}{!}{
				\includegraphics{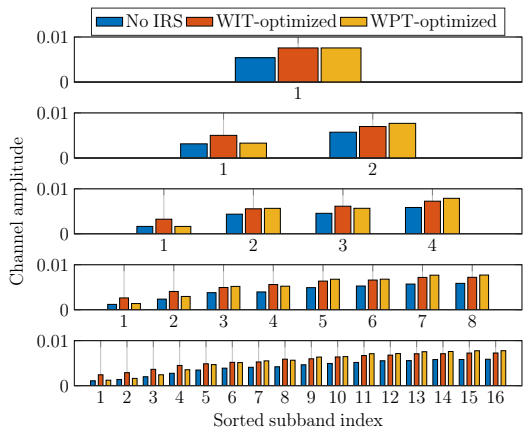}
			}
			\caption{Sorted equivalent subchannel amplitude with and without IRS versus $N$ for $M=1$, $L=100$, $\sigma_n^2=\SI{-40}{\dBm}$, $B=\SI{10}{\MHz}$ and $d_{\mathrm{H}}=d_{\mathrm{V}}=\SI{2}{\meter}$.}
			\label{fi:channel_amplitude}
		\end{figure}

		Fig.~\ref{fi:channel_amplitude} reveals how IRS influences the sorted equivalent subchannel amplitude for one channel realization. Due to the flexible subchannel design enabled by passive beamforming, the optimal amplitude distribution for WIT and WPT are dissimilar. Under the specified configuration, the WPT-optimized IRS aligns the strong subbands to exploit the rectifier nonlinearity. On the other hand, the WIT-optimized IRS provides a fair gain over all subchannels when $L$ is sufficiently large. This is reminiscent of the WF scheme at high SNR, but is realized by channel alignment instead of resource allocation. Nevertheless, the amplitude of modulated and multisine waveforms remains approximately unchanged when adding an IRS (the plots are not attached). In other words, IRS has a subtle impact on the waveform design.

		\begin{figure}[!t]
			\centering
			\subfloat[R-E region\label{fi:re_subband}]{
				\resizebox{0.45\columnwidth}{!}{
					\includegraphics{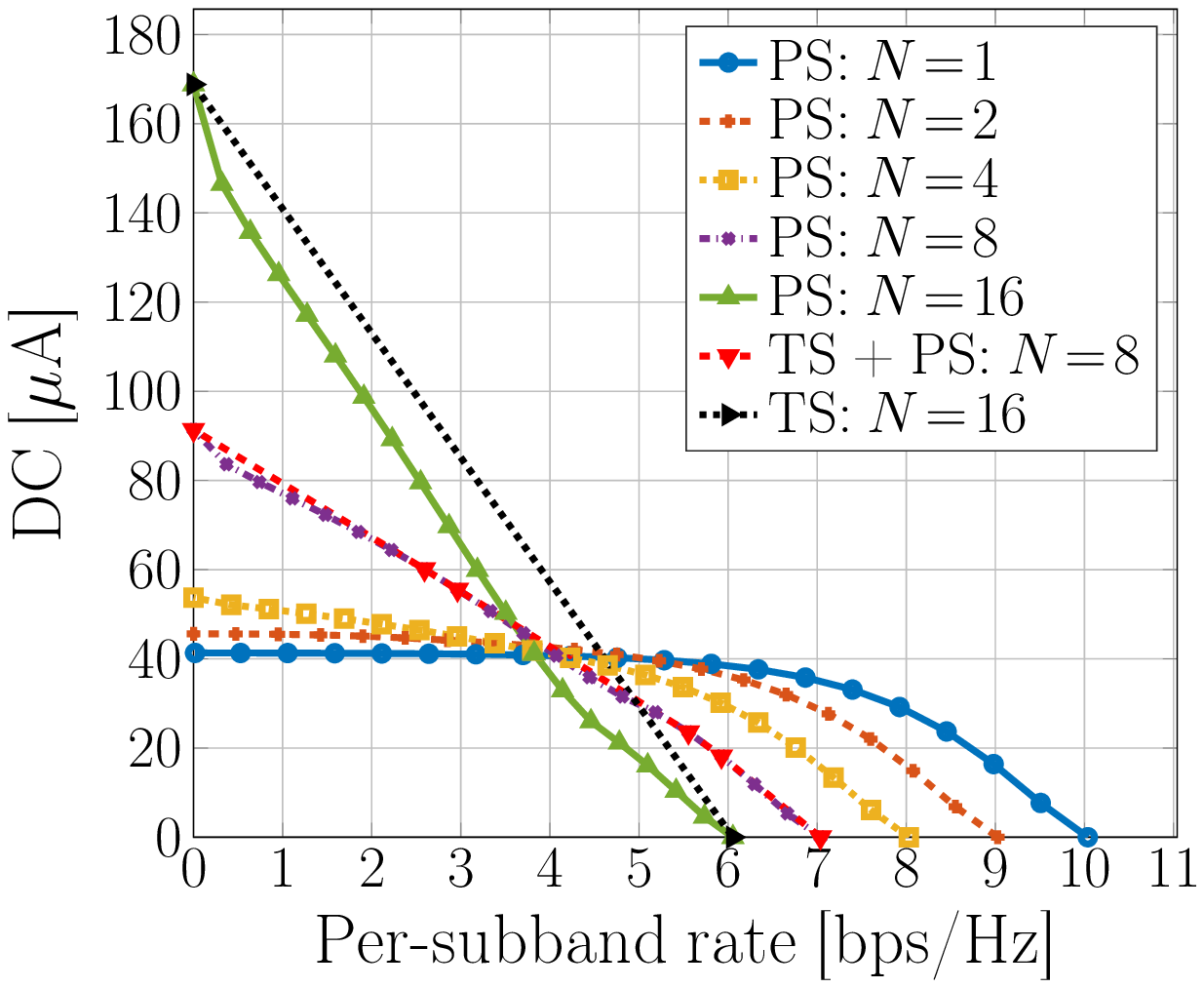}
				}
			}
			\subfloat[WPT waveform amplitude\label{fi:waveform_subband}]{
				\resizebox{0.45\columnwidth}{!}{
					\includegraphics{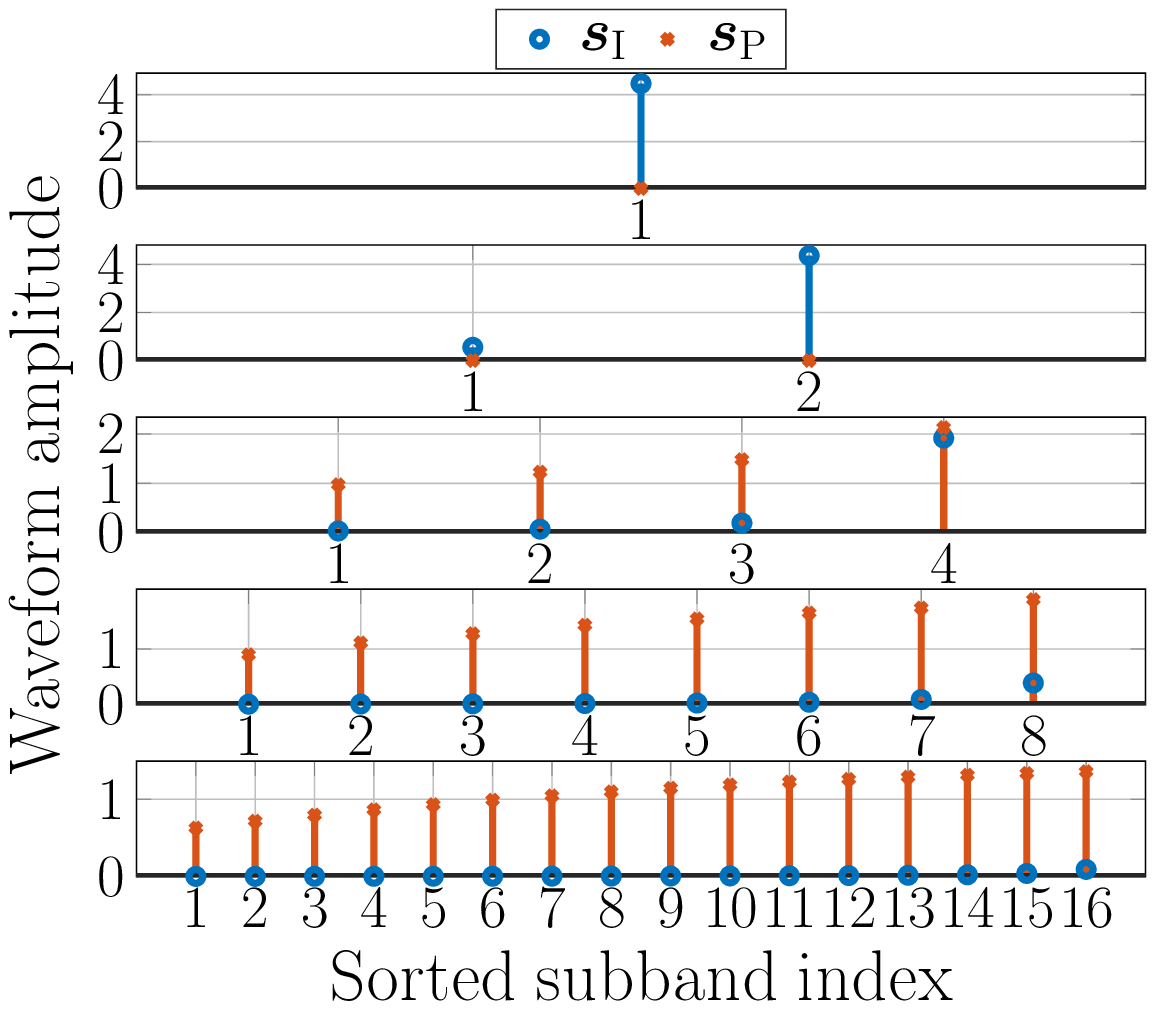}
				}
			}
			\caption{Average R-E region and WPT waveform amplitude versus $N$ for $M=1$, $L=20$, $\sigma_n^2=\SI{-40}{\dBm}$, $B=\SI{1}{\MHz}$ and $d_{\mathrm{H}}=d_{\mathrm{V}}=\SI{2}{\meter}$.}
		\end{figure}

		Fig.~\subref*{fi:re_subband} illustrates the average R-E region versus the number of subband $N$. First, it is observed that increasing $N$ reduces the per-subband rate but boosts the harvested energy. This is because less power is allocated to each subband but more balanced DC terms are introduced by frequency coupling to boost the harvested energy. On the other hand, Fig.~\subref*{fi:waveform_subband} presents the sorted modulated/multisine amplitude $\boldsymbol{s}_{\mathrm{I/P}}$ for WPT. It demonstrates that a dedicated multisine waveform is unnecessary for a small $N$ but is required for a large $N$. This observation origins from the rectifier nonlinearity. Although both waveforms have equivalent second-order DC terms \eqref{eq:y_I2} and \eqref{eq:y_P2}, for the fourth-order terms \eqref{eq:y_I4} and \eqref{eq:y_P4}, the modulated waveform has $N^2$ monomials with a modulation gain of \num{2}, while the multisine has $(2N^3+N)/3$ monomials as the components of different frequencies compensate and produce DC. Second, the R-E region is convex for $N \in \{2,4\}$ and concave-convex for $N \in \{8,16\}$, such that PS outperforms TS for a small $N$ and is outperformed for a large $N$. When $N$ is in between, the optimal strategy is a combination of both, i.e., a time sharing between the WPT point and the saddle PS SWIPT point (as denoted by the red curve in Fig.~\subref*{fi:re_subband}). When $N$ is relatively small, only modulated waveform is used at both WIT and WPT points, and one can infer that no multisine waveform is needed for the entire R-E region. It aligns with the conclusion based on the conventional linear harvester model, namely the R-E region is convex, PS outperforms TS, and dedicated power waveform is unnecessary. As $N$ becomes sufficiently large, the multisine waveform further boosts WPT and creates some concavity in the high-power region, which accounts for the superiority of TS under the nonlinear harvester model. Therefore, we conclude that the rectifier nonlinearity enlarges the R-E region by favoring a different waveform and receiving mode, both heavily depending on $N$.

		\begin{figure}[!t]
			\centering
			\subfloat[R-E region\label{fi:re_noise}]{
				\resizebox{0.45\columnwidth}{!}{
					\includegraphics{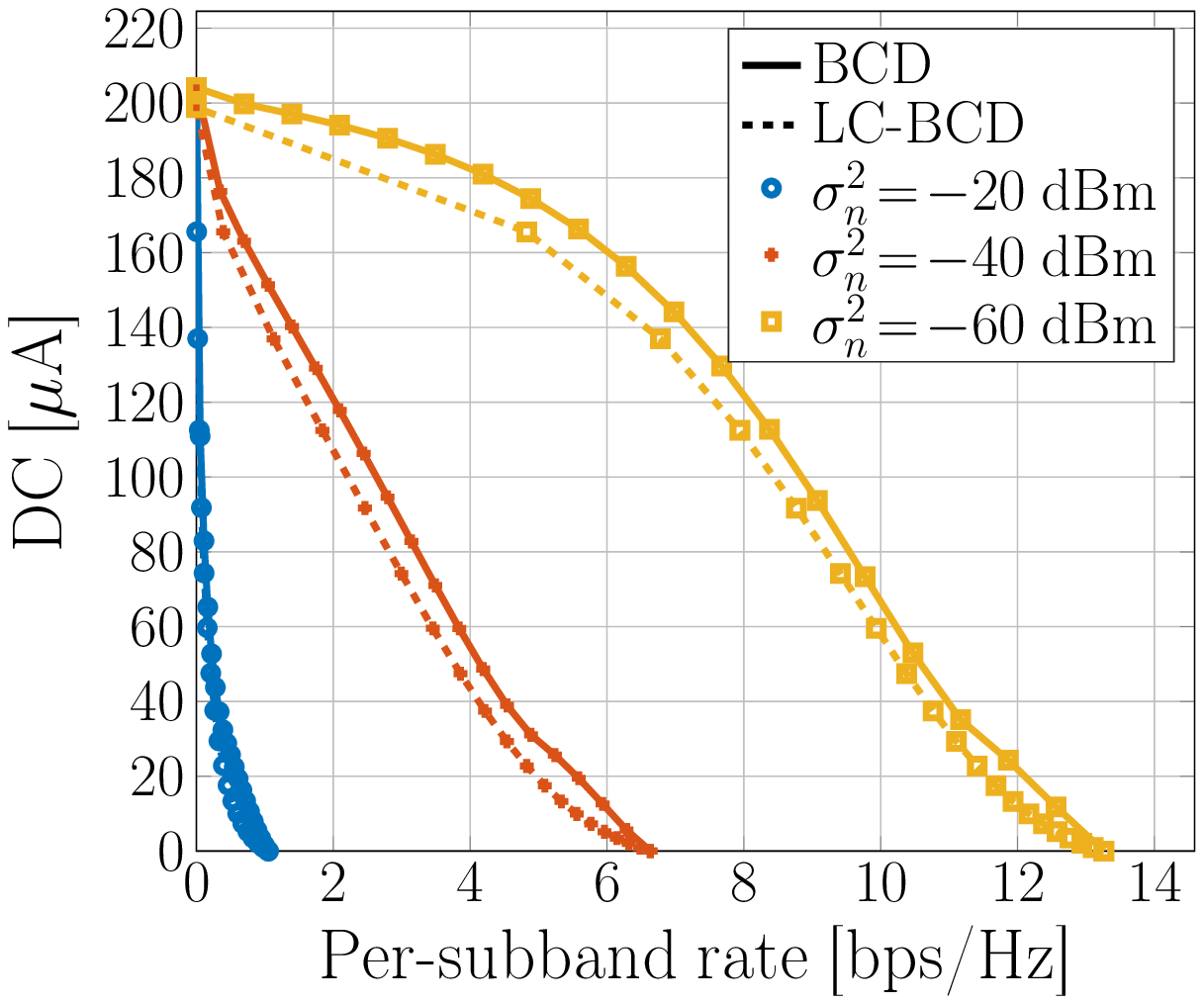}
				}
			}
			\subfloat[Splitting ratio\label{fi:splitting_ratio_noise}]{
				\resizebox{0.45\columnwidth}{!}{
					\includegraphics{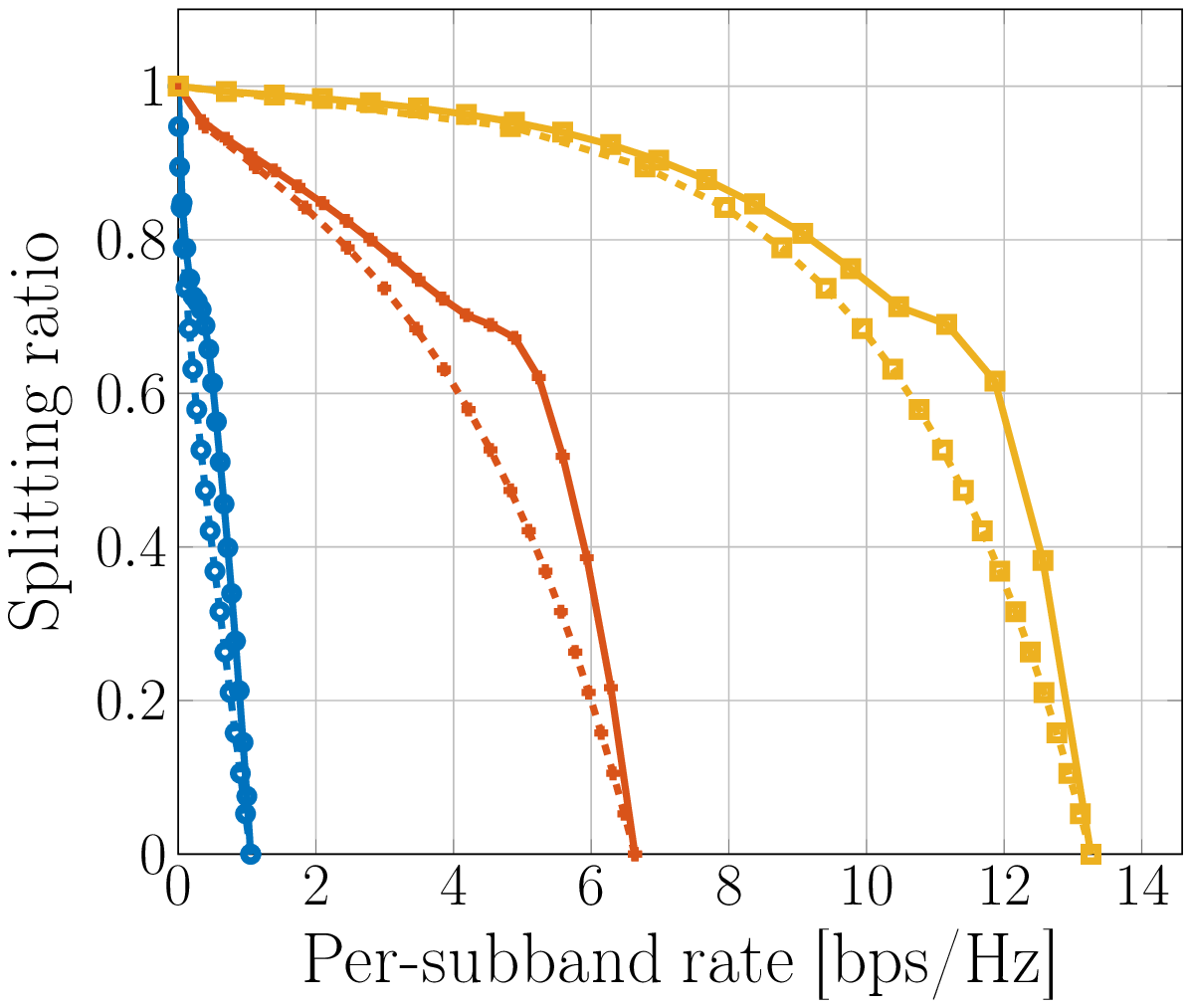}
				}
			}
			\caption{Average R-E region and splitting ratio versus $\sigma_n^2$ for $M=1$, $N=16$, $L=20$, $B=\SI{1}{\MHz}$ and $d_{\mathrm{H}}=d_{\mathrm{V}}=\SI{2}{\meter}$.}
		\end{figure}

		The average noise power influences the R-E region as shown in Fig.~\subref*{fi:re_noise}. First, we note that the R-E region is roughly concave/convex at low/high SNR such that TS/PS are preferred correspondingly. At low SNR, the power is allocated to the modulated waveform on a few strongest subbands to achieve a high rate. As the rate constraint $\bar{R}$ decreases, Algorithm~\ref{al:gp} activates more subbands that further boosts the harvested DC power because of frequency coupling and harvester nonlinearity. Second, there exists a turning point in the R-E region, especially for a low noise level ($\sigma_n^2 \le \SI{-40}{\dBm}$). The reason is that when $\bar{R}$ departs slightly from the maximum value, the algorithm tends to adjust the splitting ratio $\rho$ rather than allocate more power to the multisine waveform, since a small amplitude multisine could be inefficient for energy purpose. As $\bar{R}$ further decreases, thanks to the advantage of multisine, a superposed waveform with a small $\rho$ can outperform a modulated waveform with a large $\rho$. The result proves the benefit of superposed waveform and the necessity of joint waveform and splitting ratio optimization. Besides, the LC-BCD algorithm achieves a good balance between performance and complexity even if one-dimensional search is considered for $\delta=\rho$ from \num{0} to \num{1}.

		\begin{figure}[!t]
			\centering
			\subfloat[R-E region\label{fi:re_distance}]{
				\resizebox{0.45\columnwidth}{!}{
					\includegraphics{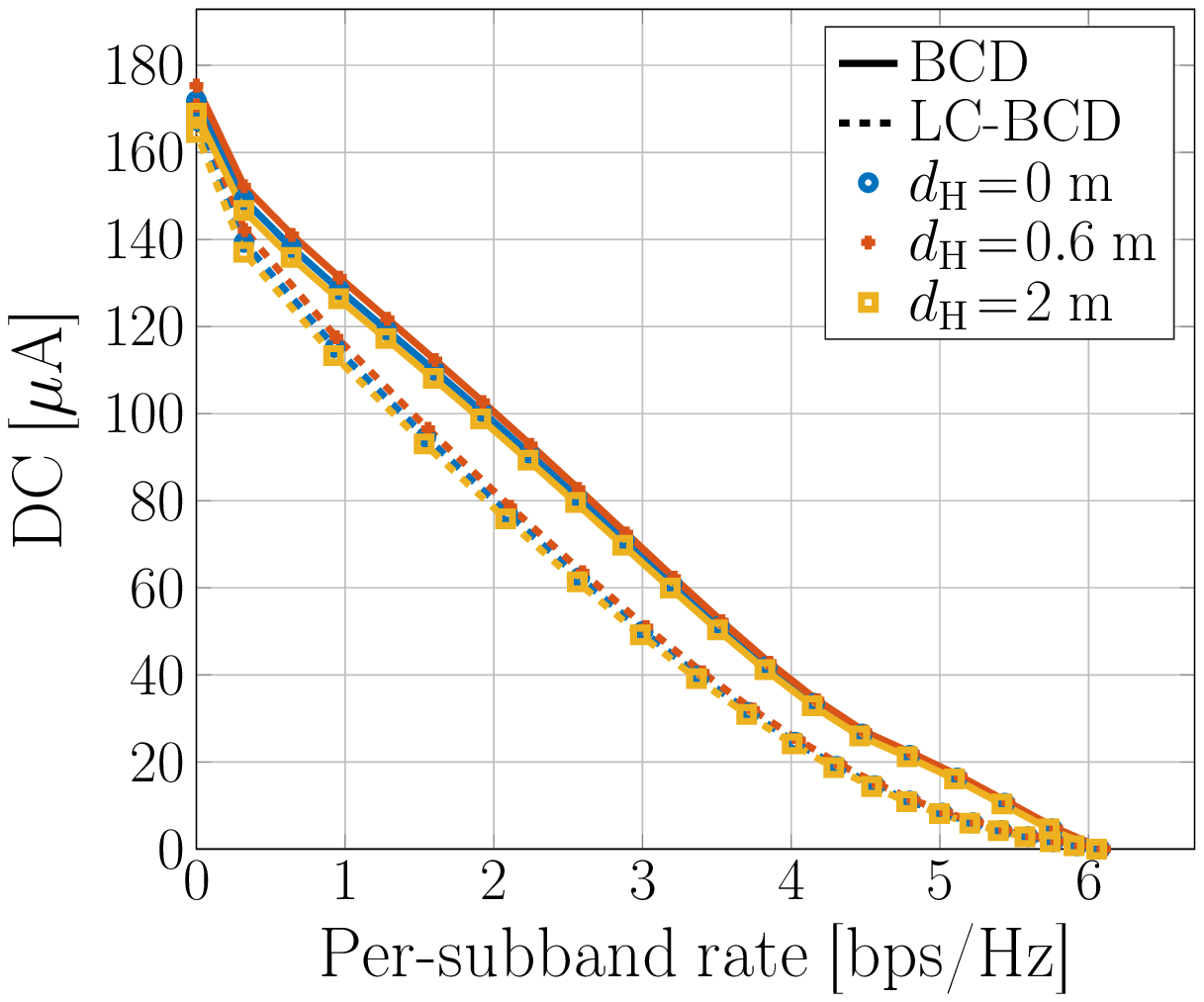}
				}
			}
			\subfloat[Path loss product\label{fi:path_loss}]{
				\resizebox{0.45\columnwidth}{!}{
					\includegraphics{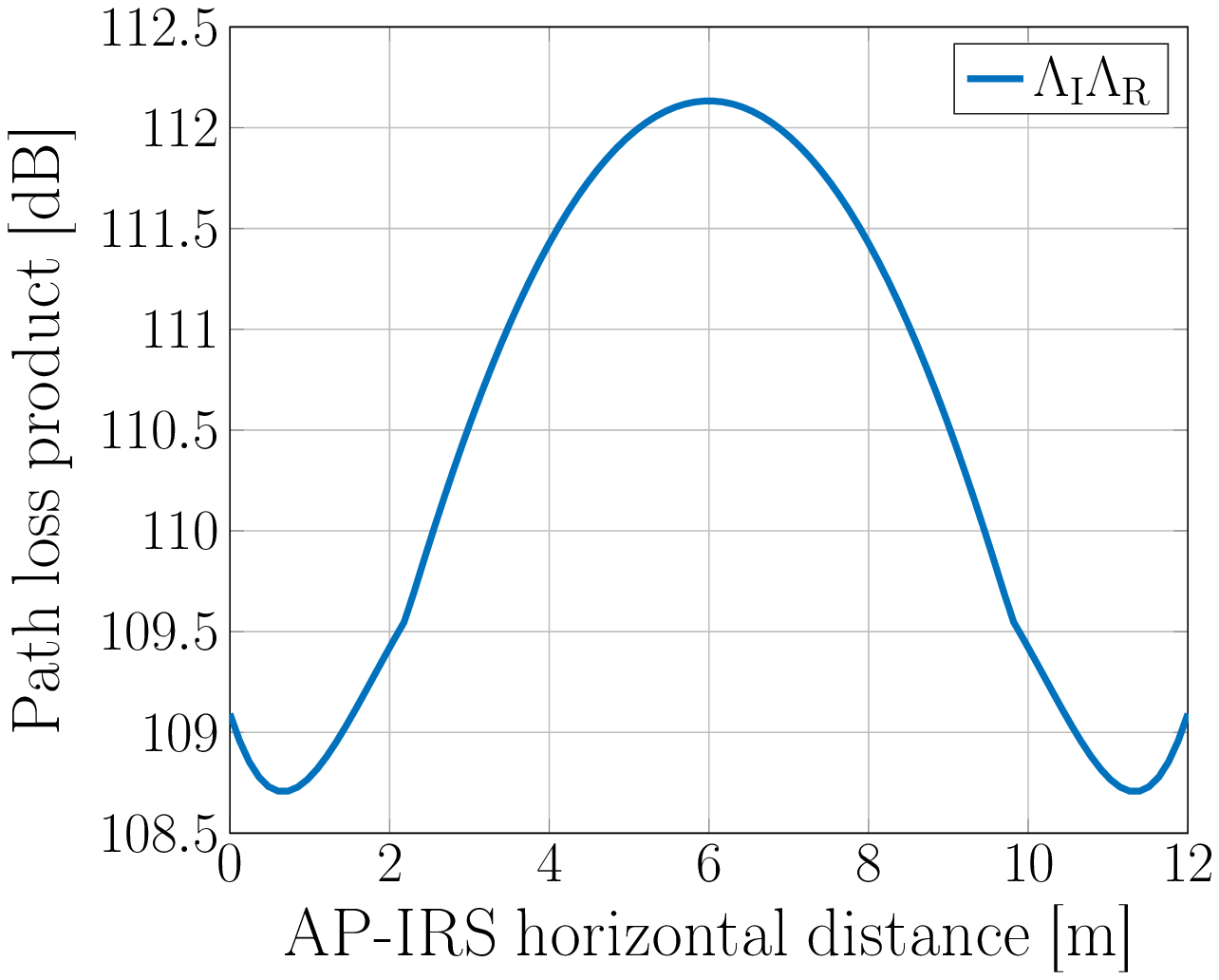}
				}
			}
			\caption{Average R-E region and path loss versus $d_{\mathrm{H}}$ for $M=1$, $N=16$, $L=20$, $\sigma_n^2=\SI{-40}{\dBm}$, $B=\SI{1}{\MHz}$ and $d_{\mathrm{V}}=\SI{2}{\meter}$.}
		\end{figure}

		In Fig.~\subref*{fi:re_distance}, we compare the average R-E region achieved by different AP-IRS horizontal distance $d_{\mathrm{H}}$. Different from the active Amplify-and-Forward (AF) relay that favors midpoint development \cite{Li2017}, the IRS should be placed close to either the AP or the UE based on the product path loss model that applies to finite-size element reflection \cite{Ozdogan2020,Tang2021}. Moreover, there exist two optimal IRS coordinates around $d_{\mathrm{H}}=0.6$ and \SI{11.4}{\meter} that minimize the path loss product $\Lambda_{\mathrm{I}}\Lambda_R$ and maximize the R-E tradeoff. It suggests that equipping the AP with an IRS can potentially extend the operation range of SWIPT systems.

		\begin{figure}[!t]
			\centering
			\subfloat[R-E region\label{fi:re_tx}]{
				\resizebox{0.45\columnwidth}{!}{
					\includegraphics{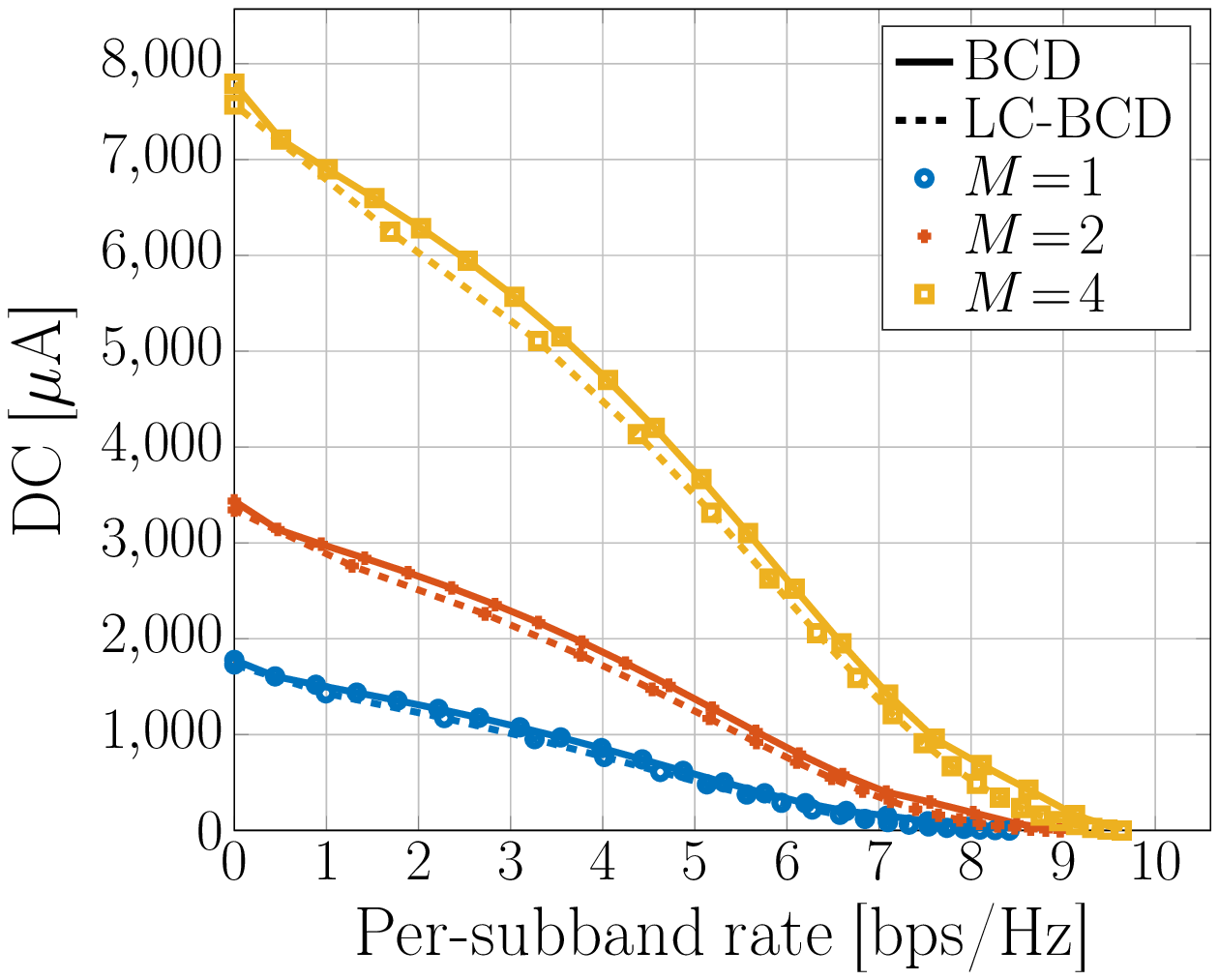}
				}
			}
			\subfloat[WIT SNR and WPT DC\label{fi:scaling_tx}]{
				\resizebox{0.45\columnwidth}{!}{
					\includegraphics{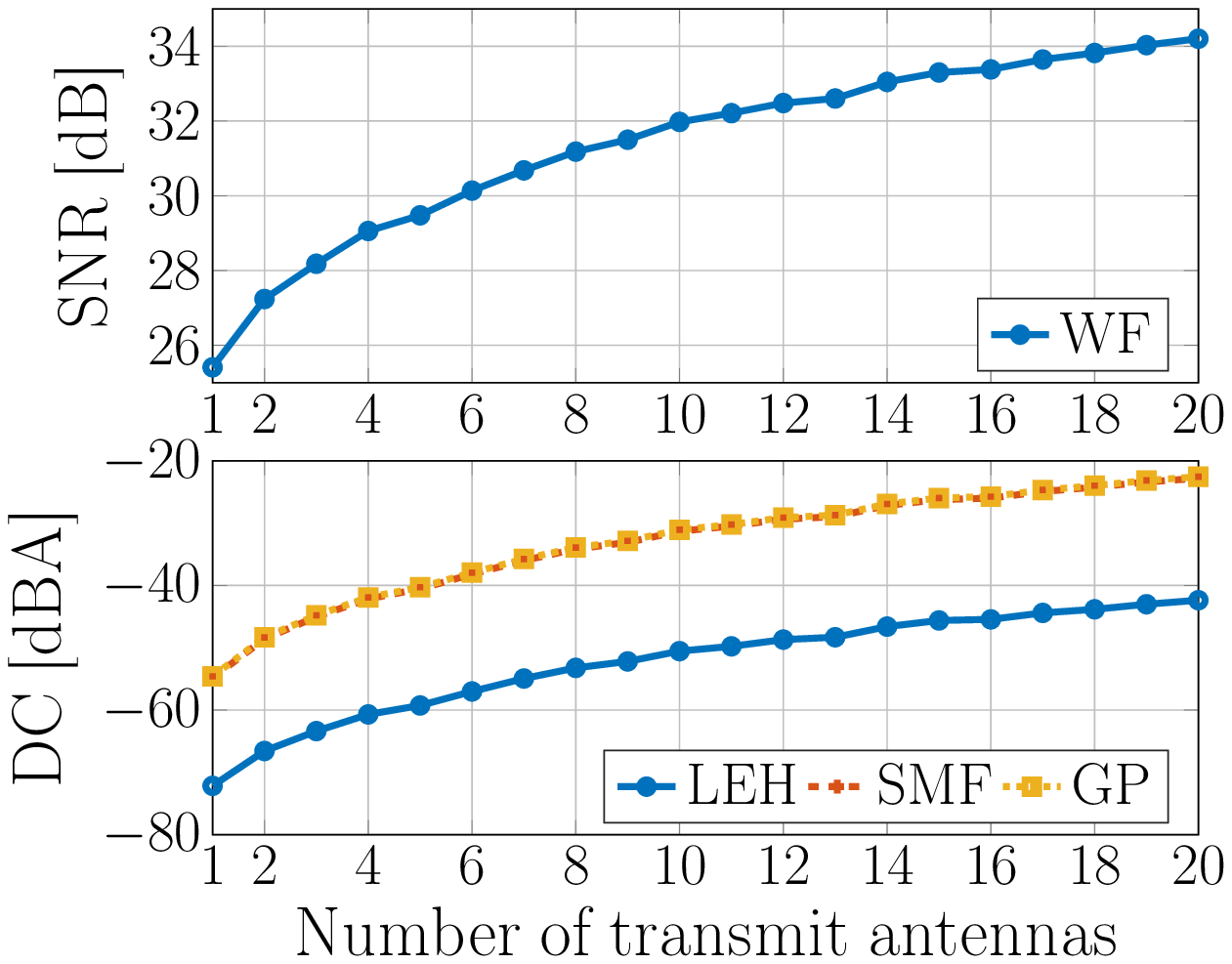}
				}
			}
			\caption{Average R-E region, WIT SNR and WPT DC versus $M$ for $N=16$, $L=20$, $\sigma_n^2=\SI{-40}{\dBm}$, $B=\SI{1}{\MHz}$, $d_{\mathrm{H}}=d_{\mathrm{V}}=\SI{0.2}{\meter}$.}
		\end{figure}

		\begin{figure}[!t]
			\centering
			\subfloat[R-E region\label{fi:re_reflector}]{
				\resizebox{0.425\columnwidth}{!}{
					\includegraphics{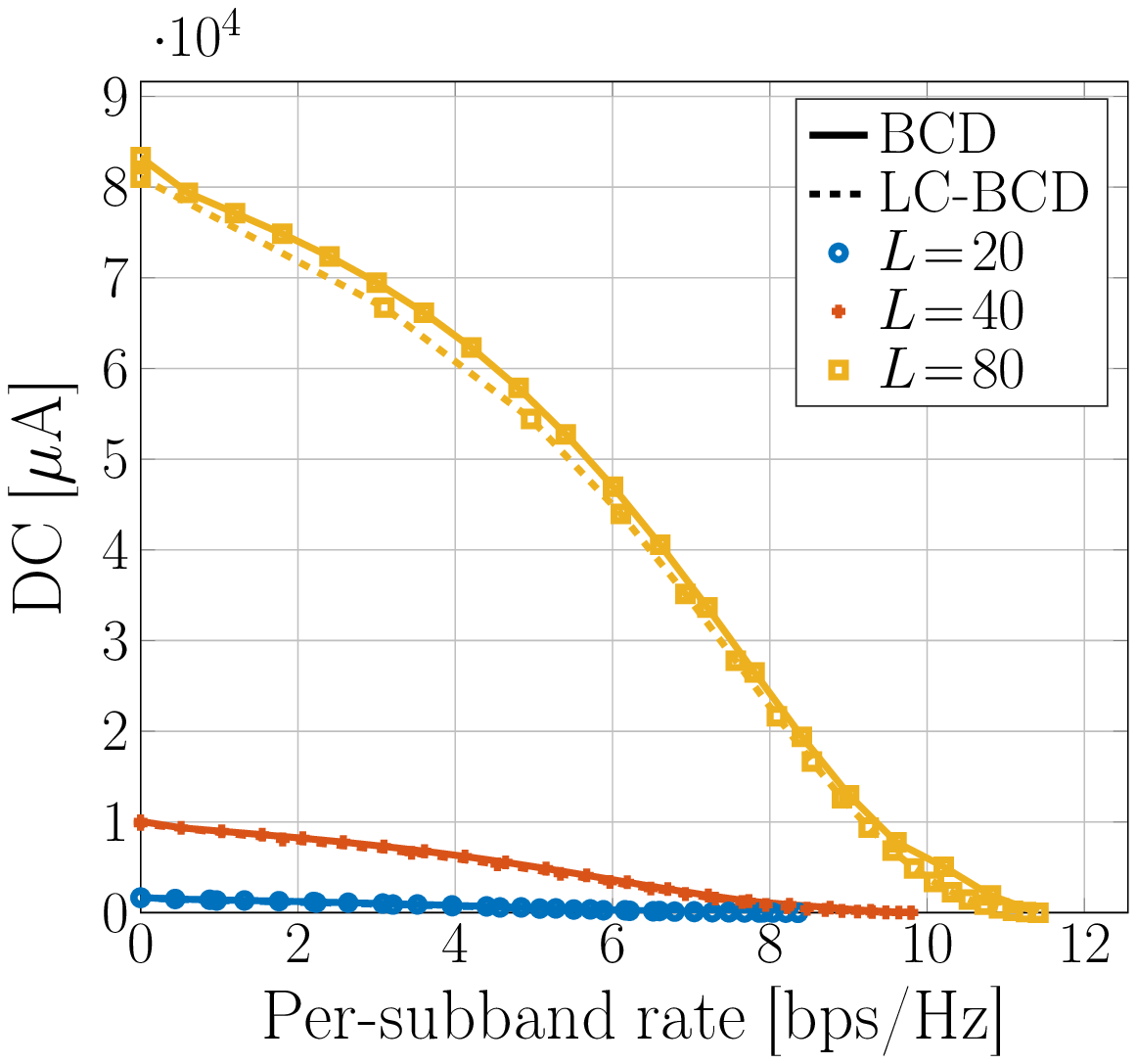}
				}
			}
			\subfloat[WIT SNR and WPT DC\label{fi:scaling_reflector}]{
				\resizebox{0.475\columnwidth}{!}{
					\includegraphics{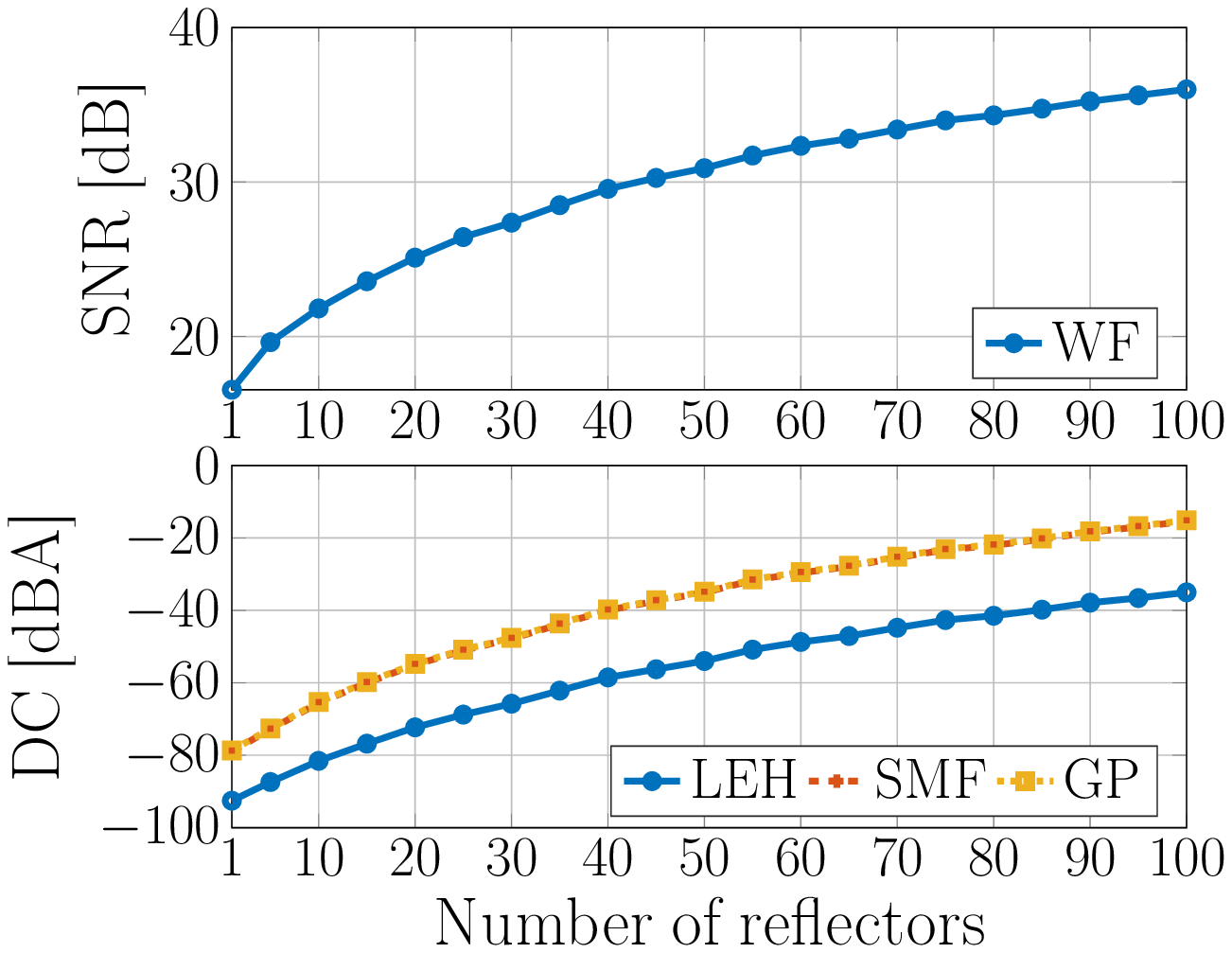}
				}
			}
			\caption{Average R-E region, WIT SNR and WPT DC versus $L$ for $M=1$, $N=16$, $\sigma_n^2=\SI{-40}{\dBm}$, $B=\SI{1}{\MHz}$ and $d_{\mathrm{H}}=d_{\mathrm{V}}=\SI{0.2}{\meter}$.}
		\end{figure}

		The impacts of the number of transmit antennas $M$ and the IRS elements $L$ on the R-E behavior are revealed in Figs.~\subref*{fi:re_tx} and \subref*{fi:re_reflector}. First, it is observed that adding either active or passive elements can improve the equivalent SNR, which produces a nearly concave R-E region and favors the PS receiver. Second, the conventional Linear Energy Harvester (LEH) model leads to a power-inefficient design. To investigate the performance loss, we truncate the DC objective function \eqref{eq:z} at $n_0=2$ such that (i) in the passive beamforming problem, $z(\boldsymbol{\Phi}) = {\beta_2}{\rho}(t_{\mathrm{I},0}+t_{\mathrm{P},0})/2$ and no SCA is required; (ii) in the waveform design problem, the WPT-optimal strategy is the adaptive single sinewave that allocates all power to the multisine at the strongest subband \cite{Clerckx2016a}. As shown in Figs.~\subref*{fi:scaling_tx} and \subref*{fi:scaling_reflector}, those conventional designs do not exploit the harvester nonlinearity and end up with a nearly \SI{20}{\dBA} gap compared to the nonlinear model-based SMF and GP designs. Third, doubling $M$ brings a \SI{3}{\dB} gain at the output SNR and a \SI{12}{\dBA} increase at the harvested DC, which verified that active beamforming has an array gain of $M$ \cite{Tse2005} with power scaling order $M^2$ under the truncated nonlinear harvester model \cite{Clerckx2016a,Clerckx2018b}. Fourth, when the IRS is very close to the AP or UE, doubling $L$ can bring a \SI{6}{\dB} gain at the output SNR and a \SI{24}{\dBA} increase at the harvested DC. From the perspective of WIT, it suggests that passive beamforming can reach an array gain of $L^2$, as indicated by \cite{Wu2019}. An interpretation is that the IRS coherently combines the incoming signal with a receive array gain $L$, then performs an equal gain reflection with a transmit array gain $L$. From the perspective of WPT, it suggests that passive beamforming comes with a power scaling order $L^4$ under the truncated nonlinear harvester model. We then verify this novel observation in a simplified case where the power is uniformly allocated over multisine, all channels are frequency-flat, and $L$ is sufficiently large such that the direct channel becomes negligible. Let $X$ be the cascaded small-scale fading coefficient. The DC in such case reduces to
		\begin{equation}
			z = \beta_2 \Lambda_{\mathrm{R}}^2 \Lambda_{\mathrm{I}}^2 \lvert X \rvert^2 L^2 P + \beta_4 \frac{2N^2 + 1}{2N} \Lambda_{\mathrm{R}}^4 \Lambda_{\mathrm{I}}^4 \lvert X \rvert^4 L^4 P^2,
		\end{equation}
		which scales quartically with $L$. Compared with active antennas, IRS elements achieve higher array gain and power scaling order, but a very large $L$ is required to compensate the double fading of the auxiliary link. These observations demonstrate the R-E benefit of passive beamforming and emphasize the importance of accounting for the harvester nonlinearity in the waveform and beamforming design.

		\begin{figure}[!t]
			\centering
			\subfloat[$B=\SI{1}{\MHz}$\label{fi:re_irs_1mhz}]{
				\resizebox{0.45\columnwidth}{!}{
					\includegraphics{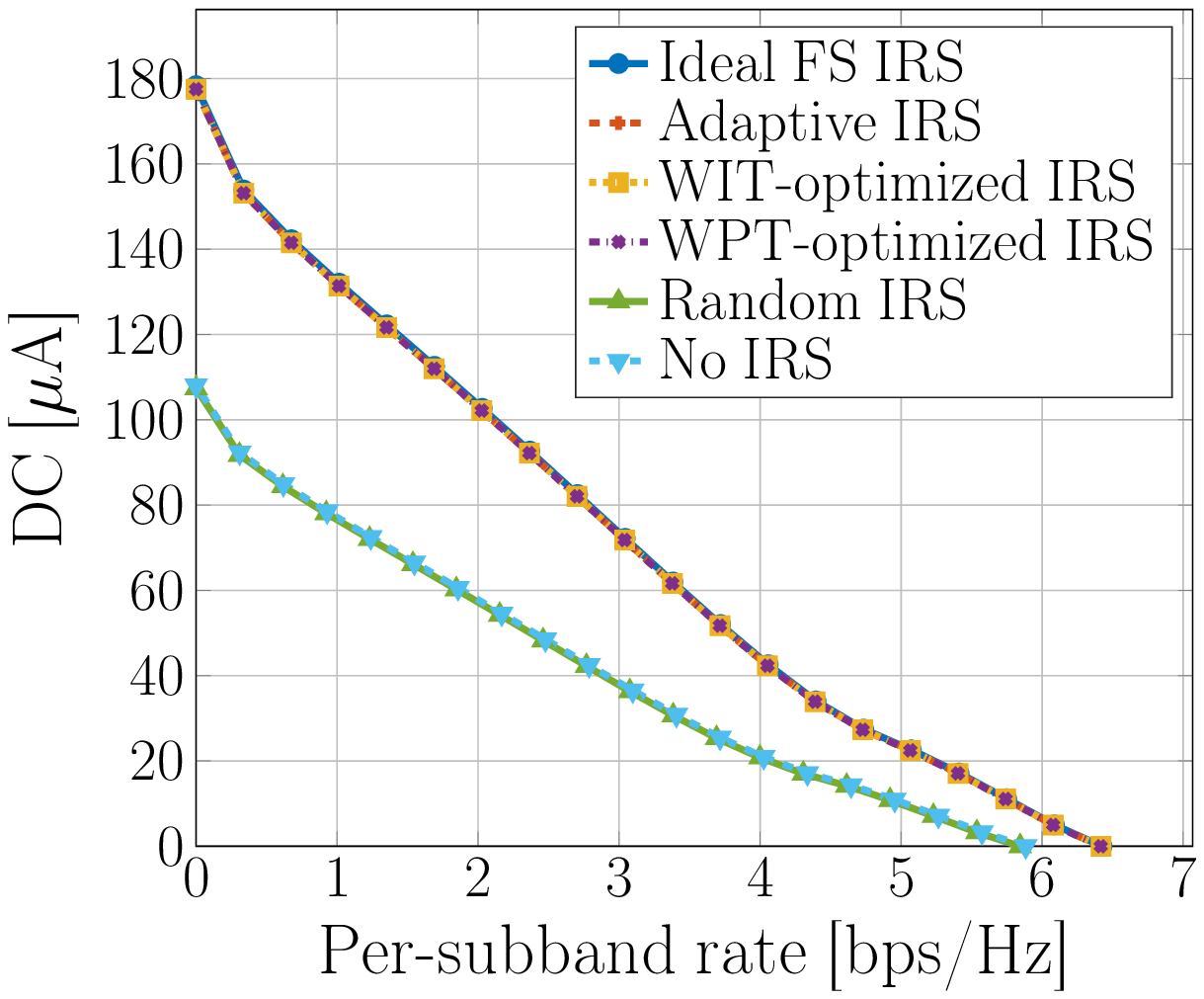}
				}
			}
			\subfloat[$B=\SI{10}{\MHz}$\label{fi:re_irs_10mhz}]{
				\resizebox{0.45\columnwidth}{!}{
					\includegraphics{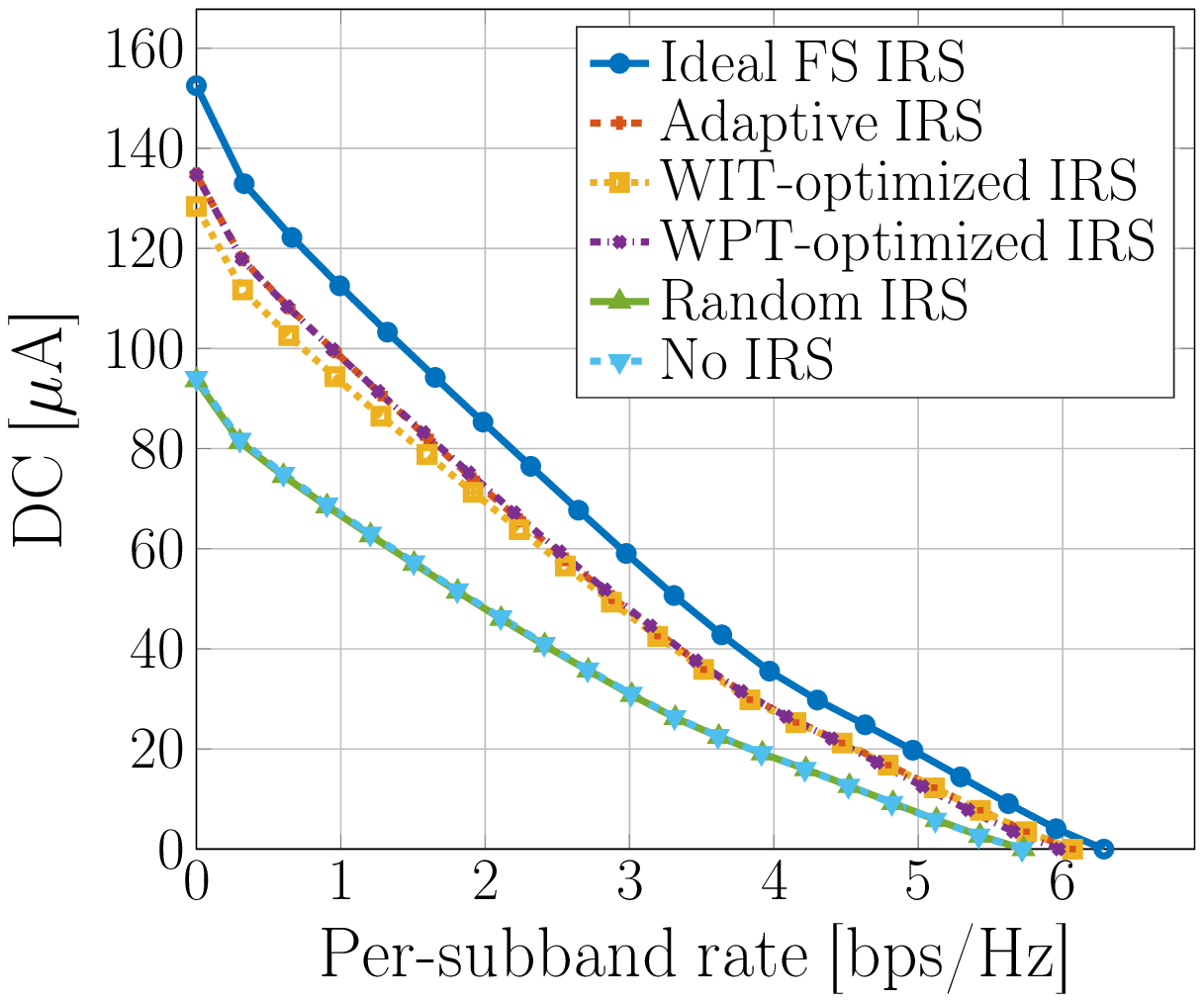}
				}
			}
			\caption{Average R-E region for ideal, adaptive, fixed and no IRS versus $B$ for $M=1$, $N=16$, $L=20$, $\sigma_n^2=\SI{-40}{\dBm}$ and $d_{\mathrm{H}}=d_{\mathrm{V}}=\SI{2}{\meter}$.}
		\end{figure}

		Figs.~\subref*{fi:re_irs_1mhz} and \subref*{fi:re_irs_10mhz} explore the R-E region with different IRS strategies for narrowband and broadband SWIPT. The ideal Frequency-Selective (FS) IRS assumes the reflection coefficient of each element is independent and controllable at different frequencies. The adaptive IRS adjusts the passive beamforming for different R-E points by Algorithm~\ref{al:sca}. The WIT/WPT-optimized IRS is retrieved by Algorithm~\ref{al:m_sca} then fixed for the whole R-E region. The random IRS models the phase shift of all elements as i.i.d. uniform random variables over $[0, 2\pi)$. First, random IRS and no IRS perform worse than other schemes since no passive beamforming is exploited. Their R-E boundaries coincide as the antenna mode reflection of the random IRS is canceled out after averaging. Second, when the bandwidth is small, the performance of ideal, adaptive, and WIT/WPT-optimized IRS are similar; when the bandwidth is large, the adaptive IRS outperforms the WIT/WPT-optimized IRS but is outperformed by the ideal FS IRS. In the former case, the subband responses are close to each other such that the tradeoff in Remark~\ref{re:subband_tradeoff} becomes insignificant, and the auxiliary link can be roughly maximized at all subbands. It suggests that for narrowband SWIPT, the optimal passive beamforming for any R-E point is optimal for the whole R-E region, and the corresponding composite channel and active precoder are also optimal for the whole R-E region. Hence, the achievable R-E region is obtained by optimizing the waveform amplitude and splitting ratio. On the other hand, since the channel frequency selectivity affects the performance of the information decoder and energy harvester differently, the optimal IRS reflection coefficient varies at different R-E tradeoffs points for broadband SWIPT. As shown in Fig.~\ref{fi:channel_amplitude}, the subchannel amplification can be either spread evenly to improve the rate at high SNR, or focused on a few strongest subbands to boost the output DC, thanks to adaptive passive beamforming.

		\begin{figure}[!t]
			\centering
			\subfloat[Imperfect cascaded CSIT\label{fi:re_csi}]{
				\resizebox{0.45\columnwidth}{!}{
					\includegraphics{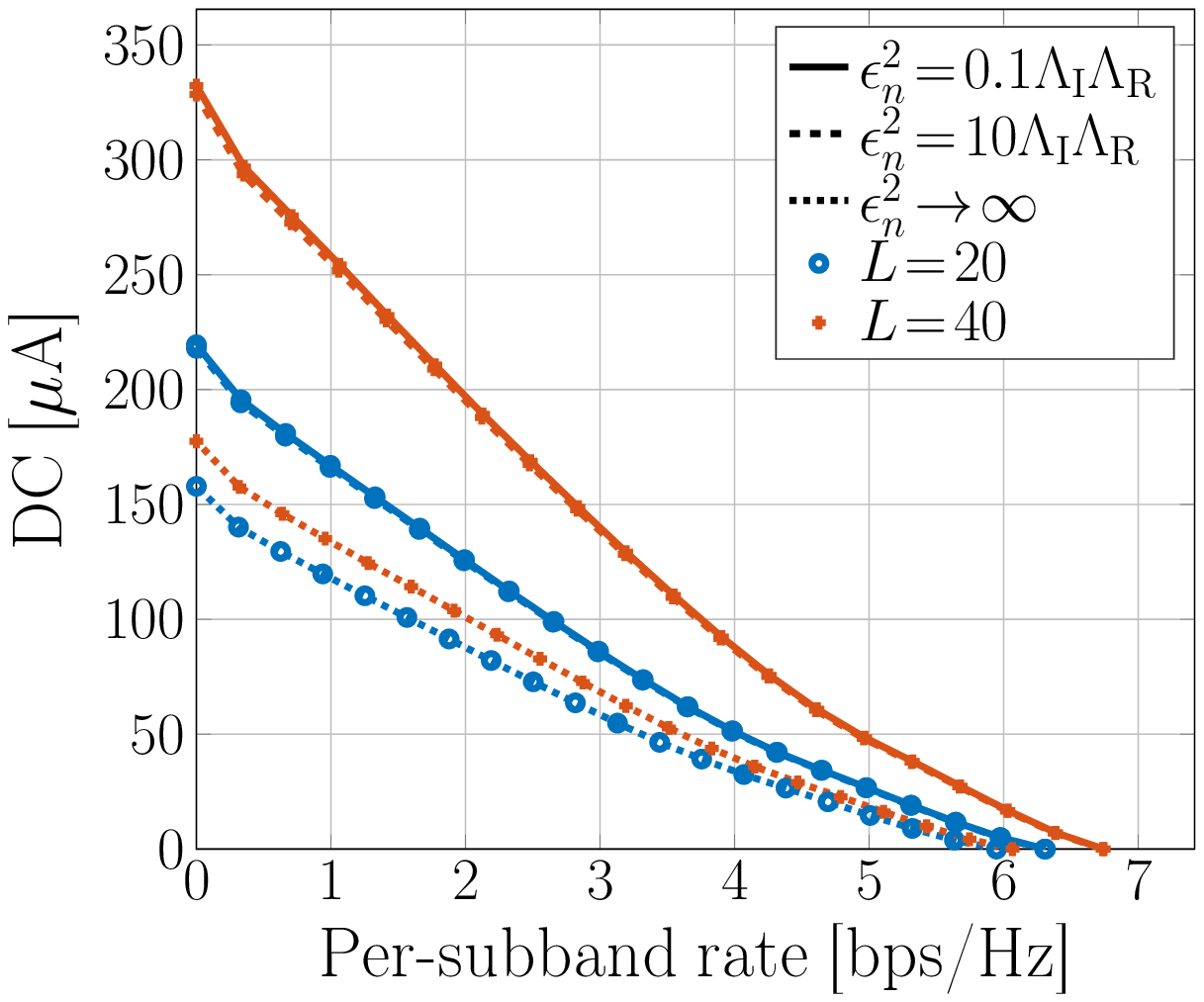}
				}
			}
			\subfloat[Quantized IRS\label{fi:re_quantization}]{
				\resizebox{0.45\columnwidth}{!}{
					\includegraphics{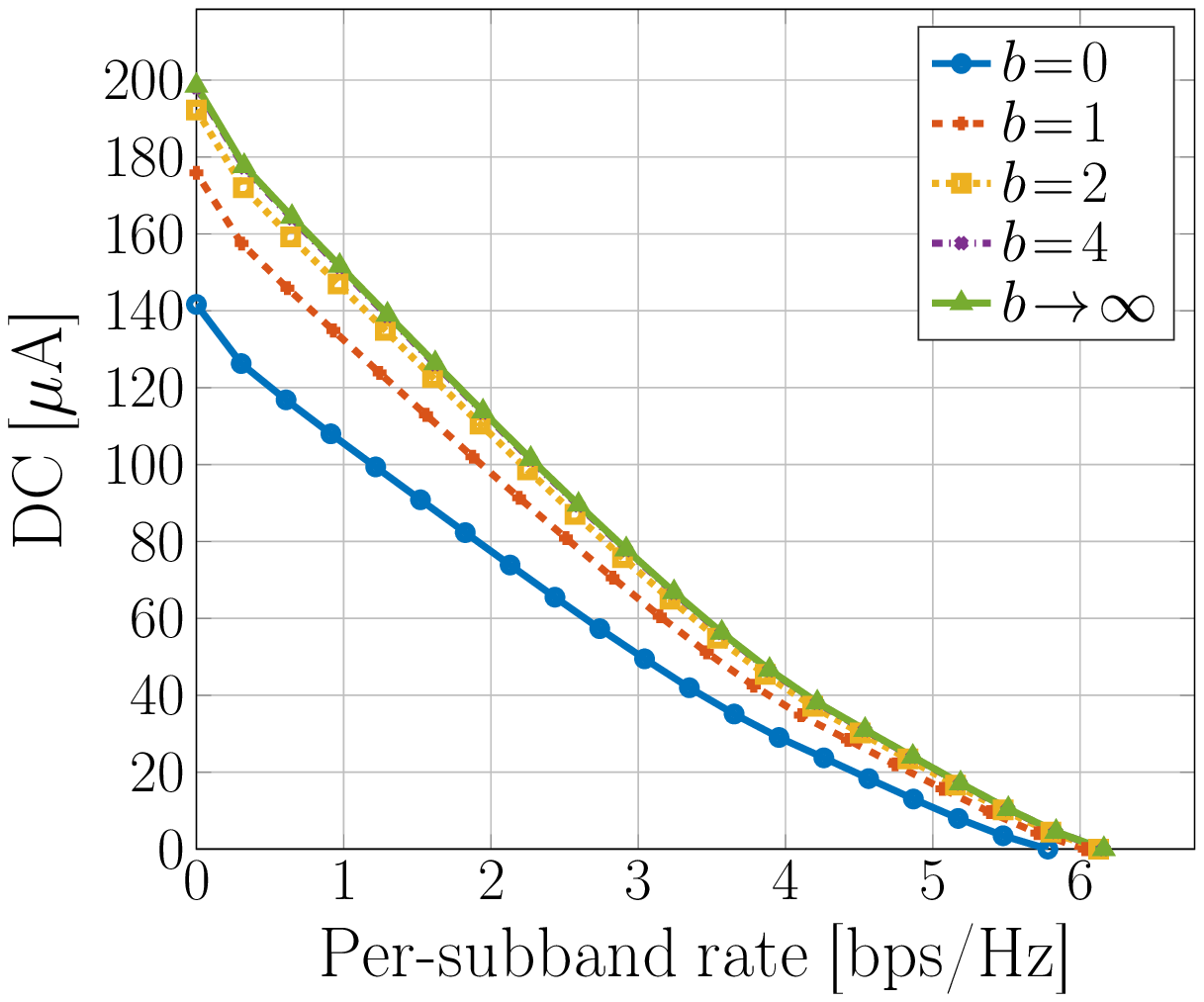}
				}
			}
			\caption{Average R-E region with imperfect cascaded CSIT and quantized IRS for $M=1$, $N=16$, $L=20$, $\sigma_n^2=\SI{-40}{\dBm}$, $B=\SI{10}{\MHz}$ and $d_{\mathrm{H}}=d_{\mathrm{V}}=\SI{2}{\meter}$. $\epsilon_{n}=0$ and $\epsilon_{n}=\infty$ correspond respectively to perfect CSIT and no CSIT (and random IRS); $b=0$ and $b \to \infty$ correspond respectively to no IRS and continuous IRS.}
		\end{figure}

		We then explore the impacts of imperfect cascaded CSIT and quantized IRS on the R-E performance. Due to the general lack of RF-chains at the IRS, it can be challenging to acquire accurate cascaded CSIT on a short-term basis. We assume the cascaded channel at subband $n$ is
		\begin{equation}
			\boldsymbol{V}_{n} = \hat{\boldsymbol{V}}_{n} + \tilde{\boldsymbol{V}}_{n},
		\end{equation}
		where $\hat{\boldsymbol{V}}_{n}$ is the estimated cascaded CSIT and $\tilde{\boldsymbol{V}}_{n}$ is the estimation error with entries following i.i.d. CSCG distribution $\mathcal{CN}(0, \epsilon_{n}^2)$.\footnote{Note that the subchannel responses are correlated but the estimations can be independent.} Figure~\subref*{fi:re_csi} shows that the proposed passive beamforming Algorithm~\ref{al:sca} is robust to cascaded CSIT inaccuracy for broadband SWIPT with different $L$. On the other hand, since the practical reflection coefficient depends on the available element impedances, we consider a discrete IRS codebook $\mathcal{C}_\phi = \{e^{j 2 \pi i / 2^b} \mid i = 1, \dots, 2^b\}$ and uniformly quantize the continuous reflection coefficients obtained by Algorithm~\ref{al:bcd} to reduce the circuit complexity and control overhead.\footnote{This relax-then-quantize approach can bring notable performance loss compared with direct optimization over the discrete phase shift set, especially for a small $b$ (i.e., low-resolution IRS) \cite{Wu2020c}.} Figure~\subref*{fi:re_quantization} suggests that even $b=1$ (i.e., two-state reflection) brings considerable R-E gain over the benchmark scheme without IRS, and the performance gap between $b=4$ and unquantized IRS is negligible. These observations demonstrate the advantage of the proposed joint waveform, active and passive beamforming design in practical IRS-aided SWIPT systems.
	\end{section}

	\begin{section}{Conclusion and Future Works}\label{se:conclusion_and_future_works}
		This paper investigated the R-E tradeoff of a single user employing practical receiving strategies in an IRS-aided multi-carrier MISO SWIPT system. Uniquely, we considered the joint waveform, active and passive beamforming design under rectifier nonlinearity to maximize the achievable R-E region. A three-stage BCD algorithm was proposed to solve the problem. In the first stage, the IRS phase shift was obtained by the SCA technique and eigen decomposition. In the second and third stages, the active precoder was derived in closed form, and the waveform amplitude and splitting ratio were optimized by the GP method. We also proposed and combined closed-form adaptive waveform schemes with a modified passive beamforming strategy to formulate a low-complexity BCD algorithm that achieves a good balance between performance and complexity. Numerical results revealed significant R-E gains by modeling harvester nonlinearity in the IRS-aided SWIPT design. Unlike active antennas, IRS elements cannot be designed independently across frequencies, but can integrate coherent combining and equal gain transmission to enable constructive reflection and flexible subchannel design. Compared to the conventional no-IRS system, the IRS mainly affects the effective channel instead of the waveform design.

        One particular unanswered question of this paper is how to design waveform, active and passive beamforming in a multi-user multi-carrier IRS-aided SWIPT system. Also, harvester saturation effect and practical IRS models with amplitude-phase coupling \cite{Abeywickrama2020}, angle-dependent reflection \cite{Tang2021}, frequency-dependent reflection, and/or partially/fully-connected architecture \cite{Shen2020a} could be considered in future works.
	\end{section}

	\begin{appendix}
		\begin{subsection}{Proof of Proposition~\ref{pr:relaxation}}\label{ap:relaxation}
			For any feasible $\boldsymbol{\Phi}$ to problem~\eqref{op:irs}, $\mathrm{tr}(\boldsymbol{\Phi})=L+1$ always holds because of the modulus constraint~\eqref{co:irs_modulus}. Therefore, we add a constant term $-\mathrm{tr}(\boldsymbol{\Phi})$ to \eqref{ob:irs} and recast problem~\eqref{op:irs} as
			\begin{maxi!}
				{\scriptstyle{\boldsymbol{\Phi}}}{-\mathrm{tr}(\boldsymbol{\Phi})+\tilde{z}(\boldsymbol{\Phi})}{\label{op:irs_equivalent}}{\label{ob:irs_equivalent}}
				\addConstraint{R(\boldsymbol{\Phi}) \ge \bar{R}}\label{co:irs_equivalent_rate}
				\addConstraint{\mathrm{diag}^{-1}(\boldsymbol{\Phi})=\boldsymbol{1}}\label{co:irs_equivalent_modulus}
				\addConstraint{\boldsymbol{\Phi}\succeq{\boldsymbol{0}}}\label{co:irs_equivalent_sd}
				\addConstraint{\mathrm{rank}(\boldsymbol{\Phi})=1.\label{co:irs_equivalent_rank}}
			\end{maxi!}
			By applying SDR, problem~\eqref{ob:irs_equivalent}--\eqref{co:irs_equivalent_sd} is convex w.r.t. $\boldsymbol{\Phi}$ and satisfies the Slater's condition \cite{Boyd2004}, and strong duality holds. The corresponding Lagrangian function at iteration $i$ is given by \eqref{eq:lagrangian}, where $\mu$, $\boldsymbol{\nu}$, $\boldsymbol{\Upsilon}$ denote respectively the scalar, vector and matrix Lagrange multiplier associated with constraint~\eqref{co:irs_equivalent_rate}, \eqref{co:irs_equivalent_modulus} and \eqref{co:irs_equivalent_sd}, and $\zeta$ collects all terms irrelevant to $\boldsymbol{\Phi}^{(i)}$.
			\begin{figure*}[!b]
				\hrule
				\begin{align}\label{eq:lagrangian}
					\mathcal{L}
					& = \mathrm{tr}\left(\boldsymbol{\Phi}^{(i)}\right) - \frac{1}{2} \beta_2 \rho \mathrm{tr}\Bigl(
							(\boldsymbol{C}_{\mathrm{I},0} + \boldsymbol{C}_{\mathrm{P},0}) \boldsymbol{\Phi}^{(i)}
						\Bigr) - \frac{3}{4} \beta_4 \rho^2 \Biggl(
							2 t_{\mathrm{I},0}^{(i-1)} \mathrm{tr}\Bigl(
								\boldsymbol{C}_{\mathrm{I},0} \boldsymbol{\Phi}^{(i)}
							\Bigr) + \sum_{k=-N+1}^{N-1} (t_{\mathrm{P},k}^{(i-1)})^* \mathrm{tr}\Bigl(
								\boldsymbol{C}_{\mathrm{P},k} \boldsymbol{\Phi}^{(i)}
							\Bigr)\nonumber\\
					& \quad + 2 t_{\mathrm{P},0}^{(i-1)} \mathrm{tr}\Bigl(
							\boldsymbol{C}_{\mathrm{I},0} \boldsymbol{\Phi}^{(i)}
						\Bigr) + 2 t_{\mathrm{I},0}^{(i-1)} \mathrm{tr}\Bigl(
							\boldsymbol{C}_{\mathrm{P},0} \boldsymbol{\Phi}^{(i)}
						\Bigr)
						\Biggr) + \mu \Biggl(
						2^{\bar{R}} - \prod_{n=1}^N \biggl(
							1 + \frac{(1-\rho) \mathrm{tr}\Bigl(
								\boldsymbol{C}_n \boldsymbol{\Phi}^{(i)}
							\Bigr)}{\sigma_n^2}
						\biggr)
					\Biggr)\nonumber\\
					& \quad + \mathrm{tr}\biggl(
						\mathrm{diag}(\boldsymbol{\nu}) \odot \Bigl(
							\boldsymbol{\Phi}^{(i)} \odot \boldsymbol{I} - \boldsymbol{I}
						\Bigr)
					\biggr) - \mathrm{tr} \Bigl(
						\boldsymbol{\Upsilon} \boldsymbol{\Phi}^{(i)}
					\Bigr) + \zeta.
				\end{align}
			\end{figure*}
			The Karush–Kuhn–Tucker (KKT) conditions on the primal and dual solutions are
			\begin{subequations}
				\begin{equation}\label{eq:lagrange_multiplier}
					\mu^\star \ge 0, \boldsymbol{\Upsilon}^\star \succeq \boldsymbol{0},
				\end{equation}
				\begin{equation}\label{eq:complementary_slackness}
					\boldsymbol{\nu}^\star \odot \mathrm{diag}^{-1}(\boldsymbol{\Phi}^\star) = \boldsymbol{0}, \boldsymbol{\Upsilon}^\star \boldsymbol{\Phi}^\star = \boldsymbol{0},
				\end{equation}
				\begin{equation}\label{eq:gradient}
					\nabla_{\boldsymbol{\Phi}^\star} \mathcal{L} = 0.
				\end{equation}
			\end{subequations}
			We then derive the gradient explicitly and rewrite \eqref{eq:gradient} as
			\begin{equation}
				\boldsymbol{\Upsilon}^\star = \boldsymbol{I} - \boldsymbol{\Delta}^\star,
			\end{equation}
			where $\boldsymbol{\Delta}^\star$ is given by \eqref{eq:delta}.
			\begin{figure*}[!b]
				\hrule
				\begin{align}\label{eq:delta}
					\boldsymbol{\Delta}^\star
					& = \frac{1}{2} \beta_2 \rho (\boldsymbol{C}_{\mathrm{I},0}+\boldsymbol{C}_{\mathrm{P},0}) + \frac{3}{4} \beta_4 \rho^2
						\Biggl(
							2 t_{\mathrm{I},0}^{(i-1)} \boldsymbol{C}_{\mathrm{I},0} + \sum_{k=-N+1}^{N-1} (t_{\mathrm{P},k}^{(i-1)})^* \boldsymbol{C}_{\mathrm{P},k} + 2 t_{\mathrm{P},0}^{(i-1)} \boldsymbol{C}_{\mathrm{I},0} + 2 t_{\mathrm{I},0}^{(i-1)} \boldsymbol{C}_{\mathrm{P},0}
						\Biggr)\nonumber\\
					& \quad + \mu^\star \sum_{n=1}^N \frac{(1-\rho) \boldsymbol{C}_n}{\sigma_n^2} \prod_{n'=1,n' \ne n}^N \Biggl(
						1 + \frac{(1-\rho)\mathrm{tr}\Bigl(
							\boldsymbol{C}_{n'} \boldsymbol{\Phi}^\star
						\Bigr)}{\sigma_{n'}^2}
					\Biggr) - \mathrm{diag}(\boldsymbol{\nu^\star}).
				\end{align}
			\end{figure*}
			Note that \eqref{eq:complementary_slackness} suggests $\mathrm{rank}(\boldsymbol{\Upsilon}^\star)+\mathrm{rank}(\boldsymbol{\Phi}^\star) \le L+1$. By reusing the proof in \cite[Appendix~A]{Xu2020}, we conclude $\mathrm{rank}(\boldsymbol{\Upsilon}^\star) \ge L$. On the other hand, $\boldsymbol{\Phi}^\star$ cannot be zero matrix and $\mathrm{rank}(\boldsymbol{\Phi}^\star) \ge 1$. Therefore, any optimal solution $\boldsymbol{\Phi}^\star$ to the relaxed problem~\eqref{op:irs_equivalent} is rank-\num{1}. Due to the equivalence between \eqref{ob:irs} and \eqref{ob:irs_equivalent}, $\boldsymbol{\Phi}^\star$ is also optimal to the relaxed problem~\eqref{op:irs} and Proposition~\ref{pr:relaxation} holds.
		\end{subsection}

		\begin{subsection}{Proof of Proposition~\ref{pr:sca}}\label{ap:sca}
			The objective function \eqref{ob:irs} is non-decreasing over iterations because the solution to \eqref{ob:irs}--\eqref{co:irs_sd} at iteration $i-1$ is still feasible at iteration $i$. Also, the sequence $\{\tilde{z}(\boldsymbol{\Phi}^{(i)})\}_{i=1}^{\infty}$ is bounded above because of the unit-modulus constraint \eqref{co:irs_modulus}. Thus, Algorithm~\ref{al:sca} is guaranteed to converge. Besides, we notice that Algorithm~\ref{al:sca} is an inner approximation algorithm \cite{Marks1978a}, because $\tilde{z}(\boldsymbol{\Phi}) \le z(\boldsymbol{\Phi})$, $\partial\tilde{z}(\boldsymbol{\Phi}^{(i)})/\partial\boldsymbol{\Phi}=\partial z(\boldsymbol{\Phi}^{(i)})/\partial\boldsymbol{\Phi}$ and the approximation \eqref{eq:taylor_1}--\eqref{eq:taylor_3} are asymptotically tight as $i \to \infty$ \cite{Li2013}. Therefore, it is guaranteed to provide a local optimal $\boldsymbol{\Phi}^{\star}$ to the relaxed passive beamforming problem. According to Proposition~\ref{pr:relaxation}, $\boldsymbol{\Phi}^{\star}$ is rank-\num{1} such that $\boldsymbol{\phi}^{\star}$ can be extracted without performance loss and the local optimality inherits to the original problem~\eqref{op:original}.
		\end{subsection}

		\begin{subsection}{Proof of Proposition~\ref{pr:mrt}}\label{ap:mrt}
			From the perspective of WIT, the MRT precoder maximizes $\lvert{\boldsymbol{h}_{n}^H \boldsymbol{w}_{\mathrm{I}, n}}\rvert = \lVert{\boldsymbol{h}_{n}}\rVert s_{\mathrm{I}, n}$ and maximizes the rate \eqref{eq:R}. From the perspective of WPT, the MRT precoder maximizes $(\boldsymbol{h}_{n}^H \boldsymbol{w}_{\mathrm{I/P}, n})(\boldsymbol{h}_{n}^H \boldsymbol{w}_{\mathrm{I/P}, n})^* = \lVert{\boldsymbol{h}_{n}}\rVert^2 s_{\mathrm{I/P}, n}^2$ and maximizes the second and fourth order DC terms \eqref{eq:y_I2}--\eqref{eq:y_P4}. Therefore, MRT is the global optimal information and power precoder.
		\end{subsection}
	\end{appendix}

	\bibliographystyle{IEEEtran}
	\bibliography{IEEEabrv,library.bib}
\end{document}